\documentclass[12pt]{article}

\usepackage{amsfonts}
\usepackage{amssymb}
\usepackage{amsmath}
\usepackage{amsthm,color}

\usepackage{verbatim}
\usepackage{hyperref}

\newtheorem{theorem}{Theorem}
\newtheorem{corollary}[theorem]{Corollary}
\newtheorem{rem1}[theorem]{Remark}
\newtheorem{lemma}[theorem]{Lemma}
\newtheorem{definition}[theorem]{Definition}
\newtheorem{proposition}[theorem]{Proposition}

\newtheorem{ex1}[theorem]{Example}

\newenvironment{remark}{\begin{rem1}\rm}{\end{rem1}}
\newenvironment{example}{\begin{ex1}\rm}{\end{ex1}}

\numberwithin{equation}{section}
\numberwithin{theorem}{section}

\newcommand{\N}{\mathbb{N}}
\renewcommand{\P}{\mathbb{P}}
\newcommand{\Q}{\mathbb{Q}}
\newcommand{\R}{\mathbb{R}}

\newcommand{\W}{\mathcal{W}}
\newcommand{\olW}{\mathcal{W}}
\newcommand{\U}{\mathbb{U}}
\newcommand{\V}{\mathbb{V}}

\newcommand{\lrparen}[1]{\left(#1\right)}

\newcommand{\lrsquare}[1]{\left[#1\right]}
\newcommand{\lsquare}[1]{\left[#1\right.}
\newcommand{\rsquare}[1]{\left.#1\right]}
\newcommand{\lrcurly}[1]{\left\{#1\right\}}
\newcommand{\lcurly}[1]{\left\{#1\right.}
\newcommand{\rcurly}[1]{\left.#1\right\}}
\newcommand{\inlrparen}[1]{\left(#1\right)}

\newcommand{\inlrcurly}[1]{\left\{#1\right\}}

\newcommand{\Ft}[1]{\mathcal{F}_{#1}}

\newcommand{\Ldpshort}[2]{L^{#1}_{#2}}

\newcommand{\LdpF}[1]{\Ldpshort{p}{#1}}
\newcommand{\LdqF}[1]{\Ldpshort{q}{#1}}
\newcommand{\LdiF}[1]{\Ldpshort{\infty}{#1}}

\newcommand{\LdzF}[1]{\Ldpshort{0}{#1}}
\newcommand{\LdpK}[3]{\Ldpshort{#1}{#2}(#3)}

\newcommand{\Lp}[2]{L^{#1}({#2})}
\newcommand{\Lpshort}[2]{L^{#1}_{#2}(\R)}
\newcommand{\LpF}[1]{\Lpshort{p}{#1}}
\newcommand{\LqF}[1]{\Lpshort{q}{#1}}
\newcommand{\LiF}[1]{\Lpshort{\infty}{#1}}
\newcommand{\LoF}[1]{\Lpshort{1}{#1}}

\newcommand{\LpK}[3]{L^{#1}_{#2}(#3)}

\newcommand{\E}[1]{\mathbb{E}\lrsquare{#1}}
\newcommand{\EP}[2]{\mathbb{E}^{#1}\lrsquare{#2}}
\newcommand{\EQ}[1]{\EP{\Q}{#1}}

\newcommand{\Et}[2]{\E{\left.#1 \right| \mathcal{F}_{#2}}}
\newcommand{\EPt}[3]{\EP{#1}{\left.#2 \right| \mathcal{F}_{#3}}}
\newcommand{\EQt}[2]{\EPt{\Q}{#1}{#2}}

\newcommand{\ERt}[2]{\EPt{\R}{#1}{#2}}

\newcommand{\dQdP}{\frac{d\mathbb{Q}}{d\mathbb{P}}}
\newcommand{\dQndP}[1]{\frac{d\mathbb{Q}_{#1}}{d\mathbb{P}}}
\newcommand{\dQidP}{\dQndP{i}}

\newcommand{\genseq}[4]{\lrparen{#1_#4}_{#4=#2}^{#3}}
\newcommand{\seq}[1]{\genseq{#1}{0}{T}{t}}

\newcommand{\trans}[1]{#1^{\mathsf{T}}}
\newcommand{\transp}[1]{\trans{\lrparen{#1}}}
\newcommand{\prp}[1]{#1^{\perp}}

\newcommand{\plus}[1]{#1^+}
\newcommand{\plusp}[1]{\plus{(#1)}}

\newcommand{\diag}[1]{#1}
\newcommand{\cl}{\operatorname{cl}}
\newcommand{\co}{\operatorname{co}}
\newcommand{\recc}[1]{\operatorname{recc}\lrparen{#1}}

\newcommand{\interior}{\operatorname{int}}

\newcommand{\as}{\text{ a.s.}}

\newcommand{\Pas}{\;\P\text{-a.s.}}

\DeclareMathOperator*{\esssup}{ess\,sup}
\DeclareMathOperator*{\essinf}{ess\,inf}

\usepackage[top=1in, bottom=1in, left=1in, right=1in]{geometry}

\begin{document}

\title{A Supermartingale Relation for Multivariate Risk Measures}
\author{Zachary Feinstein \thanks{Washington University in St. Louis, Department of Electrical and Systems Engineering, St. Louis, MO 63108, USA, zfeinstein@ese.wustl.edu.} \and Birgit Rudloff \thanks{Vienna University of Economics and Business, Institute for Statistics and Mathematics, Vienna A-1020, AUT, brudloff@wu.ac.at.}}
\date{\today~(Original: October 19, 2015)}
\maketitle
\abstract{The equivalence between multiportfolio time consistency of a dynamic multivariate risk measure and
a supermartingale property is proven. Furthermore, the dual variables under which this set-valued supermartingale is a martingale are characterized as the worst-case dual variables in the dual representation of the risk measure. Examples of multivariate risk measures satisfying the supermartingale property are given. Crucial for obtaining the results are dual representations of scalarizations of set-valued dynamic risk measures, which are of independent interest in the fast growing literature on multivariate risks.
\\[.2cm]
\\[.2cm]
\\[.2cm]
}

\section{Introduction}
\label{sec_intro}

Risk measures, introduced axiomatically in the coherent case in~\cite{AD97,AD99} and generalized to the convex case in~\cite{FS02,FG02}, quantify the minimal capital requirements to cover the risk of a financial portfolio.  
For their extension to the dynamic, multi-period setting, where the evolution of information known at time $t$ is given by a  filtration $\seq{\mathcal{F}}$, it is natural to ask how the risks relate through time. This led to the definition of time consistency.  
A risk measure is time consistent if an (almost sure) ordering of risks at a specific time implies the same ordering at all earlier points in time.  This property has been studied extensively for scalar valued risks in, e.g.,~\cite{AD07,R04,DS05,RS05,BN04,FP06,CS09,CK10,AP10,FS11} for the discrete time case and \cite{FG04,D06,DPRG10} for the continuous time case.  For the purposes of this paper, we will focus on the equivalence of  time consistency and a supermartingale property, which has been studied  for coherent risk measures in \cite{AD07,B05-thesis} and for conditionally convex risk measures in \cite{FP06,BN09}. The corresponding result reads as: A sequence $\seq{\rho}$ of sensitive conditionally convex risk measures with minimal penalty function $\alpha_t^{min}$ is time consistent if and only if the process \[V_t^{\Q}(X):=\rho_t(X) +\alpha_t^{min}(\Q)\] is a $\Q-$supermartingale, i.e., for any $\Q$ with $\alpha_0^{min}(\Q)<\infty$ and every $X\in \LiF{}$,
\[
V_t^{\Q}(X)\geq \EQt{V_{s}^{\Q}(X)}{t}  \quad \Q-a.s. 
\]
for for all times $0 \leq t < s \leq T$, see \cite[Theorem~4.5]{FP06}, \cite[Theorem~2.2.2]{P07-thesis} and \cite[Theorem~11.17]{FS11} for details on the terms and notions.
This is an important characterization as it is related to the uniform Doob decomposition under constraints, see \cite[Theorem~2.4.6]{P07-thesis}, \cite{FK97}, \cite[Chapter 9]{FS11}; provides a characterization of the time consistent version of a risk measure as the cheapest ``hedge'' of any $-X\in \LiF{}$ in the sense that $X$ is hedged at the terminal time and the incremental costs of that hedging strategy at any time $t+1$ are acceptable w.r.t.\ the original risk measure at $t$, see \cite[Proposition~2.5.2]{P07-thesis}; and provides, e.g., a representation of the superhedging costs under convex trading constraints, see \cite[Section~4.2]{P07-thesis}.

In this paper we consider set-valued or multivariate risk measures. Such risk measures and their scalarizations have been of recent interest in the literature.
They appear naturally when the random variable whose risk is to be measured is multivariate and not univariate. This is the case, e.g., when multi-asset markets (see e.g. \cite{ADM09,FMM13}) or markets with  frictions (e.g. transaction costs \cite{JMT04,HH10,HHR10,KS09, LR11,CM13} or illiquidity \cite{WA13}) are considered, or when the components of the random vector represent different risk types or the risks of different units or agents in a group. The latter case gained in particular a lot of attention recently as it allows to study the measurement and regulation of the systemic risk of banking networks, see e.g. \cite{FRW15,ACDP15,BFFM15}. It is also relevant for solvency tests of groups of insurance companies, see \cite{HMS16}. In this multivariate setting one is usually interested in the allocations of the capital charges to the different units, agents, risk types, assets, or currencies. A set-valued risk measure provides these quantities as it assigns to a random vector the set of all capital allocations that compensate its risk.
Set-valued risk measures have been studied in a single period framework in, e.g.,~\cite{JMT04,HR08,HH10,HHR10,CM13}.  Dynamic, multi-period set-valued risk measures have been studied in, e.g.,~\cite{FR12,FR12b,FR13-survey,FR14-alg,TL12}. We will mostly follow the setting and notation of~\cite{FR12,FR12b} in this paper.  
In the dynamic multivariate case, a version of time consistency, called multiportfolio time consistency, is used.  This property has been shown to be equivalent to a set-valued version of many of the same properties that time consistency is equivalent to in the scalar case.

In this paper we will show that the equivalence of time consistency and a supermartingale property that is well known for scalar dynamic risk measures can be proven in the multivariate case as well. That is, we will show that multiportfolio time consistency of a normalized conditionally convex dynamic risk measures $\seq{R}$ satisfying certain continuity properties is equivalent to
the set-valued stochastic process
\[ \V_t^{(\Q,w)}(X) := \cl\lrsquare{R_t(X) + \alpha_t(\Q,w)} \]
being a supermartingale, that is, satisfying for every $(\Q,w) \in \W_t$ and all $X \in \LdpK{p}{}{\R^d}$
\[\V_t^{(\Q,w)}(X) \subseteq \cl\EQt{\V_s^{(\Q,w_t^s(\Q,w))}(X)}{t}\] 
for all times $0 \leq t < s \leq T$. Here, $\alpha_t$ denotes the set-valued penalty function of $R_t$. All terms and notions will be made precise in the main sections of the paper.
Note, that as usual in the multivariate setting, one can work on general $\LdpK{p}{}{\R^d}$ spaces in contrast to the $\LiF{}$ framework of the scalar setting.  This is due to the fact that one works with upper sets and not just their boundaries, so the scalar problems stemming from the usage of the essential supremum disappear in the present framework for multivariate risks. The technique to prove the supermartingale results differs drastically from the scalar case.  The proofs rely crucially on dual representations of (conditional) scalarizations of set-valued risk measures which are deduced within this paper. These results are of independent interest in the very 
active field of research on scalar (but static) multivariate risk measures, see e.g. \cite{ADM09,FMM13,FS06,K09,Sc04} and also \cite{ACDP15,BFFM15,FRW15} for applications to systemic risk. Results on the dynamic case are thus also influential in these areas.  All of these results are proven in the Appendix.

Set-valued sub- and supermartingales are defined e.g.\ in \cite[Chapter 3]{M05}, \cite[Chapter 4]{LOK02}, and \cite[Chapter 8]{HP97} for random closed sets and a Doob decomposition is given in \cite[Chapter 4.7]{LOK02}.  However, due to the ordering relation used here, where smaller risk corresponds to a larger set which thus contains smaller capital requirements, our notion of a supermartingale corresponds to a submartingale in those works.

Characterizations of set-valued supermartingales are highly desirable as set-valued stochastic processes play an important role in many fields of research, e.g.\ in statistics \cite{WZ91,L10}, random set theory \cite{M05}, and for stochastic differential inclusions \cite{K13}, with applications to economics and control theory \cite{KMM03a,KMM03b}.
Furthermore, the obtained result can be seen as a stepping stone towards future research on a set-valued uniform Doob decomposition as well as on hedging of multivariable claims w.r.t.\ multivariate risk measures.

The paper is structured as follows. In Section~\ref{sec_prelim} we will review properties of dynamic multivariate risk measures from \cite{FR12,FR12b} and present a new dual representation for such risk measures.  In Section~\ref{sec_supermtg} we provide results on the equivalence of multiportfolio time consistency and a set-valued supermartingale property for convex and coherent multivariate risk measures.  Finally, in Section~\ref{sec_supermtg-cc} 
we present the main results by extending the results of Section~\ref{sec_supermtg}, focusing on conditionally convex and conditionally coherent multivariate risk measures.
We will provide examples of risk measures satisfying these supermartingale properties. The proofs and intermediate results are collected in the Appendix.

\section{Set-valued dynamic risk measures}
\label{sec_prelim}
In this section we will present notation, definitions, and simple results about duality and multiportfolio time consistency for set-valued dynamic risk measures which can be derived from~\cite{FR12,FR12b}.

We will work with a general filtered probability space $(\Omega,\Ft{},\seq{\mathcal{F}},\P)$ satisfying the usual conditions with $\Ft{T} = \Ft{}$.  This setting allows for either a discrete time $\{0,1,...,T\}$ or continuous time $[0,T]$ framework. Consider the linear spaces $\LdpF{t} := \Lp{p}{\Omega,\Ft{t},\P;\R^d}$ for any $p \in [1,\infty]$ and denote $\LdpF{} := \LdpF{T}$, where $\LdpF{t}$ is the linear space of the equivalence classes of $\Ft{t}$-measurable random vectors $X: \Omega \to \R^d$ with $\|X\|_p^p = \int_\Omega |X(\omega)| d\P < \infty$ for $p < \infty$ and $\|X\|_{\infty} = \esssup_{\omega \in \Omega} |X(\omega)| < \infty$ for $p = \infty$, where $|\cdot|$ denotes an arbitrary norm in $\R^d$.  We will consider the norm topology on $\LdpF{}$ for $p \in [1,\infty)$ and the weak* topology on $\LdiF{}$ when $p = \infty$.  The closure operator $\cl$ is taken as the topological closure throughout this work.

Let $\LdpK{p}{t}{D_t} := \{Z \in \LdpF{t} \; | \; Z \in D_t \Pas\}$ denote the set of random vectors in $\LdpF{t}$ which take values in $D_t$ $\P$-a.s.  Additionally, throughout this paper we sometimes need to distinguish the spaces of random vectors from those of random variables.  To do so, let us denote the linear space of equivalence classes of random variables with finite $p$-norm $X: \Omega \to \R$ by $\LpF{t} := \Lp{p}{\Omega,\Ft{t},\P;\R}$. We will use the following notation for the different types of multiplication: multiplication between a random variable $\lambda \in \LiF{}$ and a set of random vectors $D \subseteq \LdpF{}$ is defined elementwise $\lambda D = \{\lambda Y \; | \; Y \in D\} \subseteq \LdpF{}$ with $(\lambda Y)(\omega) = \lambda(\omega)Y(\omega)$; the componentwise multiplication between random vectors is denoted by $\diag{X}Y := \transp{X_1 Y_1,...,X_d Y_d} \in \LdzF{}$ for $X,Y \in \LdzF{}$.

Throughout we will use the notation $\LdpF{t,+} := \inlrcurly{X \in \LdpF{t} \; | \; X \in \R^d_+ \Pas}$ to denote the set of $\Ft{t}$-measurable random vectors with $\P$-a.s.\ non-negative components.  Similarly, we will define $\LdpF{t,++} := \inlrcurly{X \in \LdpF{t} \; | \; X \in \R^d_{++} \Pas}$ as those $\Ft{t}$-measurable random vectors which are $\P$-a.s.\ strictly positive.  As with the prior notation, we will define $\LdpF{+} := \LdpF{T,+}$ and $\LdpF{++} := \LdpF{T,++}$.  (In)equalities between random vectors (resp. variables) are understood componentwise in the $\P$-a.s.\ sense.  The set $\LdpF{+}$ defines an ordering on the space of random vectors: $Y \geq X$ for $X,Y \in \LdpF{}$ when $Y - X \in \LdpF{+}$.

In financial contexts, a random vector $X \in \LdpF{t}$ represents a portfolio in the sense that component $X_i$ gives the number of units of asset $i \in \{1,...,d\}$ held at time $t$.  Thus, we consider portfolios in ``physical units'' of assets instead of the value of the portfolio in some num\'{e}raire.  This framework was used and discussed in, e.g., \cite{K99,S04,KS09}.

For risk measurement purposes, fix $m \in \{1,...,d\}$ of the assets to be eligible for covering the risk of a portfolio.  Without loss of generality we will assume the eligible assets are the first $m$ assets; these first assets may correspond to, e.g., reserve currencies such as US Dollars, Euros, and Yen.  We will denote by $M := \R^m \times \{0\}^{d-m}$ the subspace of eligible assets.  The set of eligible portfolios is given by $M_t := \LdpK{p}{t}{M}$; this is a closed (weak* closed if $p = \infty$) linear subspace of $\LdpF{t}$ (cf. Section 5.4 and Proposition 5.5.1 of \cite{KS09}).  Denote $M_{t,+} := M_t \cap \LdpF{t,+}$ to be the non-negative eligible portfolios and $M_{t,-} := -M_{t,+}$ to be the non-positive eligible portfolios. For example, in the application of systemic risk (cf. \cite{FRW15}) the eligible portfolios are often chosen so that $M = \R^d$ for $d$ banks and thus $M_t = \LdpF{t}$ for all times $t$.

We are now able to introduce the conditional risk measures as in \cite{FR12,FR12b}.  A conditional risk measure is a mapping of portfolios (i.e.\ $d$-dimensional random vectors) into the upper sets
\[\mathcal{P}\lrparen{M_t;M_{t,+}} := \lrcurly{D \subseteq M_t \; | \; D = D + M_{t,+}},\]
which is a subset of the power set $2^{M_t}$.  The output for portfolio $X$ is the set $R_t(X)$ at time $t$, which is the collection of all eligible portfolios that compensate for the risk.
\begin{definition}\cite[Definition 2.1]{FR12b}
\label{defn_conditional}
A function $R_t: \LdpF{} \to \mathcal{P}\inlrparen{M_t;M_{t,+}}$ is a \textbf{\emph{normalized (conditional) risk measure}} at time $t$ if it is
\begin{enumerate}
\item $M_t$-translative: for every $m_t \in M_t: R_t\inlrparen{X + m_t} = R_t(X) - m_t$;
\item $\LdpF{+}$-monotone: $Y \geq X$ implies $R_t(Y) \supseteq R_t(X)$;
\item finite at zero: $R_t(0) \not\in \{\emptyset,M_t\}$;
\item normalized: for every $X \in \LdpF{t}: R_t(X) = R_t(X) + R_t(0)$.
\end{enumerate}

Additionally, a conditional risk measure at time $t$ is \textbf{\emph{(conditionally) convex}} if for all $X,Y \in \LdpF{}$ and all $ 0 \leq \lambda \leq 1$ ($\lambda \in \LiF{t}$ such that $0 \leq \lambda \leq 1$)
\[R_t(\lambda X + (1-\lambda)Y) \supseteq \lambda R_t(X) + (1-\lambda)R_t(Y),\]
it is \textbf{\emph{(conditionally) positive homogeneous}} if for all $X \in \LdpF{}$,
for all $\lambda >0$ ($\lambda \in L^\infty_t(\R_{++})$)
\[R_t(\lambda X) = \lambda R_t(X),\]
and is \textbf{\emph{(conditionally) coherent}} if it is (conditionally) convex and (conditionally) positive homogeneous.

A conditional risk measure at time $t$ is \textbf{\emph{closed}} if the graph of the risk measure, \[\operatorname{graph}R_t = \lrcurly{(X,u) \in \LdpF{} \times M_t\;| \; u \in R_t(X)},\] is closed in the product topology.

A conditional risk measure at time $t$ is \textbf{\emph{convex upper continuous (c.u.c.)}} if for any closed set $D \in \mathcal{G}(M_t;M_{t,-}) := \inlrcurly{D \subseteq M_t\;|\; D = \cl\co(D + M_{t,-})}$
\[R_t^{-1}(D) := \lrcurly{X \in \LdpF{}\;|\; R_t(X) \cap D \neq \emptyset}\] 
is closed. 
\end{definition}

The properties given in Definition~\ref{defn_conditional} and their interpretations are discussed in detail in~\cite{HHR10,FR12,FR12b}. Briefly, the four axiomatic properties for a normalized risk measure guarantee the interpretation of risk as a `capital requirement', (almost sure) larger returns correspond with less risk, and the zero portfolio has `zero' risk. Convexity and coherence provide the interpretation on how diversification impacts the risk of a portfolio.
The image space of a closed convex conditional risk measure is 
\[\mathcal{G}(M_t;M_{t,+}) = \lrcurly{D \subseteq M_t\;|\; D = \cl \co \lrparen{D + M_{t,+}}}.\]
Note that any c.u.c.\ risk measure is closed and any closed risk measure has closed images.  

A \textbf{dynamic risk measure} $\seq{R}$ is a sequence of conditional risk measures and is said to have one of the properties given in Definition~\ref{defn_conditional} if for every time $t$ the conditional risk measure $R_t$ has that property.

For any risk measure $R_t$ there exists an \textbf{acceptance set} and vice versa, see Remark~2 and Proposition~2.11 in~\cite{FR12}.  For a conditional risk measure $R_t$ the associated acceptance set is defined by the portfolios that require no additional capital to cover the risk, i.e.
\[A_t := \lrcurly{X \in \LdpF{} \; | \; 0 \in R_t(X)}.\]
Given an acceptance set $A_t$, the associated conditional risk measure is defined by the eligible portfolios that, when added to the initial portfolio, make that acceptable, i.e.
\[R_t(X) := \lrcurly{u \in M_t \; | \; X + u \in A_t}.\]

\begin{example}
\label{ex:systemic-defn}
\emph{Aggregation-based risk measure}: Here we will consider an example of an aggregation-based risk measure, which in the static framework is a special case of the systemic risk measures introduced in \cite{FRW15}. We will return to this example throughout this paper so as to provide a consistent illustration of the results.  For simplicity let $p = +\infty$. Consider the full space of eligible assets for all times $t$, i.e., $M_t = \LdiF{t}$. Let $\Lambda: \R^d \to \R$ be defined by $\Lambda(x) = \sum_{i = 1}^d \lrparen{-a x_i^- + b x_i^+}$ for any $x \in \R^d$ and with $a \geq b > 0$ (where $x_i^+ := \max\{0,x_i\}$ and $x_i^- := (-x_i)^+$). This aggregation function was considered in, e.g., \cite{BC14}. Additionally, consider the (scalar) worst case acceptance set $\mathcal{A}_t := \LpK{\infty}{}{\R_+}$ for all times $t$.  The aggregation based risk measure and acceptance set are then defined for any $X \in \LdiF{}$ as
\begin{align*}
R_t^{\Lambda}(X) &:= \lrcurly{u \in M_t \; | \; \Lambda(X + u) \in \mathcal{A}_t} = \bigcap_{k = 0}^{2^d-1} \lrcurly{u \in M_t \; | \; \trans{v_k}u \geq \rho_t^{WC}(\trans{v_k}X) \as}\\
A_t^{\Lambda} &:= \Lambda^{-1}[\mathcal{A}_t] = \LdpK{\infty}{}{\Lambda^{-1}[\R_+]}.
\end{align*}
In the above, $\rho_t^{WC}: \LiF{} \to \LiF{t}$ is the scalar worst case risk measure, i.e., the risk measure with acceptance set $\mathcal{A}_t$ (see also \cite[Example~4.8]{FS11} in a static setting and \cite[Example~4]{R04} in a dynamic setting).  Additionally, the vector $v_k \in \{a,b\}^d$ for each $k = 0,1,...,2^d-1$ is defined component-wise as $v_{k,i} = a 1_{\{\operatorname{mod}(\lfloor k/2^{i-1} \rfloor,2) = 0\}} + b 1_{\{\operatorname{mod}(\lfloor k/2^{i-1} \rfloor,2) = 1\}}$ for $i = 1,2,...,d$.
Further, by definition, this risk measure is closed and conditionally coherent, and thus normalized as well.
\end{example}

We will define the stepped risk measure and acceptance set by restricting the domain of portfolios to the future eligible assets.  That is, for times $0 \leq t < s \leq T$, the stepped risk measure $R_{t,s}: M_s \to \mathcal{P}(M_t;M_{t,+})$ is defined by $R_{t,s}(X) := R_t(X)$ for any $X \in M_s$ and the stepped acceptance set is defined by $A_{t,s} := \lrcurly{X \in M_s \; | \; 0 \in R_t(X)} = A_t \cap M_s$.  We refer to \cite[Appendix C]{FR12b} for a detailed discussion of the stepped risk measures.

\subsection{Dual representation}
\label{subsection_dual}
In this section we will present the robust representation for conditional risk measures.  In \cite{FR12,FR12b}, a dual representation utilizing the negative convex conjugate, as defined in~\cite{H09}, was given. For this paper, the main results simplify when using the dual representation w.r.t.\ the positive convex conjugate introduced in~\cite{setOPsurvey}. This dual representation will be deduced below.
To provide these results, we will first define the Minkowski subtraction for sets $A,B \subseteq M_t$ by
\[A -^. B = \lrcurly{m \in M_t \; | \; B + \{m\} \subseteq A}.\]
The remainder of the setting is identical to that of \cite{FR12,FR12b}, which we will quickly summarize.
Denote the space of $d$-dimensional probability measures that are absolutely continuous with respect to the physical measure $\P$ by $\mathcal{M}$.  For notational purposes let $\mathcal{M}^e \subseteq \mathcal{M}$ be the set of vector probability measures equivalent to $\P$.  We will consider a $\P$-a.s.\ version of the $\Q$-conditional expectation (for $\Q := \transp{\Q_1,...,\Q_d} \in \mathcal{M}$).  Let
\[\EQt{X}{t} := \Et{\diag{\xi_{t,T}(\Q)} X}{t}\] for any $X \in \LdpF{}$, where $\xi_{t,s}(\Q) =
\transp{\bar{\xi}_{t,s}(\Q_1),...,\bar{\xi}_{t,s}(\Q_d)}$ for any $0 \leq t \leq s \leq T$ with 
\[\bar{\xi}_{t,s}(\Q_i) := \begin{cases}\frac{\Et{\dQidP}{s}}{\Et{\dQidP}{t}} & \text{on } \lrcurly{\omega \in \Omega \; | \; \Et{\dQidP}{t}(\omega) > 0}\\ 1 & \text{else} \end{cases},\] 
see e.g.~\cite{CK10,FR12}.  Note that for any probability measure $\Q_i \ll \P$ it follows that $\dQidP = \bar{\xi}_{0,T}(\Q_i)$.
We will say $\Q = \P|_{\Ft{t}}$, i.e.\ $\Q$ is equal to $\P$ on $\Ft{t}$, if $\xi_{0,t}(\Q) = \transp{1,...,1} \in \R^d$.

We will denote the half-space $G_t(w)$ and the conditional ``half-space'' $\Gamma_t(w)$ in $\LdpF{t}$ with normal direction $w \in \LdqF{t}\backslash \{0\}$ by
\[G_t(w) := \lrcurly{u \in \LdpF{t}\;|\; 0 \leq \E{\trans{w}u}}, \quad\quad \Gamma_t(w) := \lrcurly{u \in \LdpF{t}\;|\; 0 \leq \trans{w}u}.\]
For ease of notation, for any $w \in \LdqF{t}\backslash \{0\}$ we write
\[G_t^M(w) := G_t(w) \cap M_t, \quad\quad \Gamma_t^M(w) := \Gamma_t(w) \cap M_t.\]

From \cite{FR12b}, we can define the set of dual variables at time $t$ to be
\begin{align*}
\olW_t &= \lrcurly{(\Q,w) \in \mathcal{M} \times \lrparen{\plus{M_{t,+}} \backslash \prp{M_t}} \;|\; w_t^T(\Q,w) \in \LdqF{+}, \;\Q= \P|_{\Ft{t}}}
\end{align*}
with \[w_t^s(\Q,w) := \diag{w} \xi_{t,s}(\Q)\] for any $0 \leq t \leq s \leq T$.  We define the positive dual cone of a cone $C \subseteq \LdpF{t}$ (in particular for $C = M_{t,+}$) by
\[C^+=\lrcurly{v \in \LdqF{t}\;|\; \forall u \in C: \E{\trans{v}u} \geq 0}\] and the orthogonal space of $M_t$ by
\[\prp{M_t} = \lrcurly{v \in \LdqF{t}\;|\; \forall u \in M_t: \E{\trans{v}u} = 0}.\]

For ease of readability, we denote in this paper the positive conjugate by $\beta$, respectively $\alpha$ for the conditionally convex case, and the negative conjugate (used in the proofs) by $-\bar\beta$ (resp. $-\bar\alpha$).  This is in contrast to the notation used in \cite{FR12,FR12b}, where $-\beta$, respectively $-\alpha$, denoted the negative conjugate. The positive and the negative conjugate functions are related to each other by $\beta_t(\Q,w) := G_t^M(w) -^. (-\bar\beta_t(\Q,w))$, respectively, $\alpha_t(\Q,w) := \Gamma_t^M(w) -^. (-\bar\alpha_t(\Q,w))$, for any $(\Q,w) \in \W_t$.

\begin{corollary}
\label{cor_dual}
A function $R_t: \LdpF{} \to \mathcal{G}(M_t;M_{t,+})$ is a \textbf{\emph{closed convex  risk measure}} if and only if
\begin{equation}
\label{Eq:dual}
R_t(X) = \bigcap_{(\Q,w) \in \olW_t} \lrsquare{\lrparen{\EQt{-X}{t} + G_t\lrparen{w}} \cap M_t -^. \beta_t(\Q,w)},
\end{equation}
where $\beta_t$ is the minimal penalty function given by
\begin{equation}
\label{min_penalty}
\beta_t(\Q,w) = \bigcap_{Y \in A_t} \lrparen{\EQt{-Y}{t} + G_t(w)} \cap M_t.
\end{equation}
$R_t$ is additionally coherent if and only if
\begin{equation*}
R_t(X) = \bigcap_{(\Q,w) \in \olW_{t}^{\max}} \lrparen{\EQt{-X}{t} + G_t\lrparen{w}} \cap M_t,
\end{equation*}
for
\begin{equation*}
\olW_{t}^{\max} = \lrcurly{(\Q,w) \in \olW_t\;|\; w_t^T(\Q,w) \in \plus{A_t}}.
\end{equation*}
\end{corollary}
\begin{proof}
The results follow from the dual representation in \cite[Theorem 2.3]{FR12b} and $\beta_t(\Q,w) := G_t^M(w) -^. (-\bar\beta_t(\Q,w))$.
\end{proof}

\begin{corollary}
\label{cor_cond_dual}
A function $R_t: \LdpF{} \to \mathcal{G}(M_t;M_{t,+})$ is a \textbf{\emph{closed conditionally convex risk measure}} if and only if
\begin{equation}
\label{cond_convex_dual}
R_t(X) = \bigcap_{(\Q,w) \in \olW_t} \lrsquare{\lrparen{\EQt{-X}{t} + \Gamma_t\lrparen{w}} \cap M_t -^. \alpha_t(\Q,w)},
\end{equation}
where $\alpha_t$ is the minimal conditional penalty function given by
\begin{equation}
\label{cond_min_penalty}
\alpha_t(\Q,w) = \bigcap_{Y \in A_t} \lrparen{\EQt{-Y}{t} + \Gamma_t(w)} \cap M_t.
\end{equation}
$R_t$ is additionally conditionally coherent if and only if
\begin{equation}
\label{conditional_coherent_dual}
R_t(X) = \bigcap_{(\Q,w) \in \olW_{t}^{\max}} \lrparen{\EQt{-X}{t} + \Gamma_t\lrparen{w}} \cap M_t.
\end{equation}
\end{corollary}

The proof of Corollary~\ref{cor_cond_dual} is much more involved than the proof of Corollary~\ref{cor_dual} and will be given in Secion~\ref{proofCordual} of the Appendix.

\begin{example}
\label{ex:systemic-dual}
\emph{Aggregation-based risk measure}: Consider again the aggregation-based risk measure constructed in Example~\ref{ex:systemic-defn}. The dual representation of this conditionally coherent risk measure is given by the set of dual variables
\[\olW_t^{\Lambda} = \lrcurly{(\Q,w) \in \W_t \; | \; w_t^T(\Q,w) \in \LdpK{1}{}{\plus{\Lambda^{-1}[\R_+]}}}.\]
As $a \geq b > 0$ in the definition of the aggregation function $\Lambda$, if $(\Q,w) \in \olW_t^{\Lambda}$ then $\Q \in \mathcal{M}^e$.  This additional property becomes important in Section~\ref{sec_supermtg-cc}. Furthermore, it holds that $(\Q,v_k) \in \olW_t^{\Lambda}$ for any probability measure $\Q \in \mathcal{M}^e$ and $v_k$ as defined in Example~\ref{ex:systemic-defn} for every $k = 0,1,...,2^{d-1}$ since $\trans{v_k}x \geq \Lambda(x)$ for any choice of $x \in \R^d$ by construction.
\end{example}

\subsection{Multiportfolio time consistency}\label{subsection_mptc}
Multiportfolio time consistency has been studied in \cite{FR12,FR12b} as a useful concept of time consistency for set-valued risk measures.  We will quickly review the definition and some of the equivalent characterizations of this property.  In particular, we will provide the cocycle condition on (positive) penalty functions as being equivalent to multiportfolio time consistency as this result will be used in the main proofs of the paper. In contrast, in \cite{FR12b} this result was shown for the negative penalty functions.

\begin{definition}\cite[Definition~2.7]{FR12b}
\label{defn_mptc}
A dynamic risk measure $\seq{R}$ is \textbf{\emph{multiportfolio time consistent}} if for all times $0 \leq t < s \leq T$, all portfolios $X\in \LdpF{}$ and all sets ${\bf Y}\subseteq \LdpF{}$ the following implication is satisfied
\begin{equation*}
  R_s(X) \subseteq \bigcup_{Y \in {\bf Y}} R_s(Y) \Rightarrow R_t(X) \subseteq \bigcup_{Y \in {\bf Y}} R_t(Y).
\end{equation*}
\end{definition}
Conceptually, a risk measure is multiportfolio time consistent if, whenever any eligible portfolio that compensates for the risk of $X$ will compensate for the risk of some portfolio $Y \in {\bf Y}$ at some time, then at any prior time the same relation holds.
\begin{theorem}\cite[Theorem 3.4]{FR12}
\label{thm_equiv_tc}
For a normalized dynamic risk measure $\seq{R}$ the following are equivalent:
\begin{enumerate}
\item \label{thm_equiv_tctc}$\seq{R}$ is multiportfolio time consistent,
\item \label{thm_equiv_recursive} $R_t$ is recursive; that is for all times $0 \leq t < s \leq T$
    \begin{equation}
    \label{recursive}
        R_t(X) = \bigcup_{Z \in R_s(X)} R_t(-Z) =: R_t(-R_s(X)).
    \end{equation}
\item $A_t = A_{t,s} + A_s$ for every time $0 \leq t < s \leq T$.
\end{enumerate}
\end{theorem}

As shown in the above theorem, multiportfolio time consistency is equivalent to a recursive relation for set-valued risk measures.  Furthermore, \cite{FR14-alg} discusses the relation between the recursive form and a set-valued version of Bellman's principle.  

In the case of discrete time $\{0,1,...,T\}$, a step size of $1$ (i.e.\ setting $s = t+1$) is sufficient to define multiportfolio time consistency and the recursive relation \eqref{recursive}.

\begin{example}
\label{ex:systemic-mptc}
\emph{Aggregation-based risk measure}: Consider again the aggregation-based risk measure constructed in Example~\ref{ex:systemic-defn}. As the acceptance sets $A_t^{\Lambda}$ are constant in time, it therefore follows that $R_t^{\Lambda}(X) = R_{t+1}^{\Lambda}(X) \cap \LdiF{t}$.  Therefore, directly by the definition, we can prove that this aggregation-based risk measure is multiportfolio time consistent since for any $0 \leq t < s \leq T$, $X \in \LdiF{}$ and ${\bf Y} \subseteq \LdiF{}$ such that $R_s^{\Lambda}(X) \subseteq \bigcup_{Y \in {\bf Y}} R_s^{\Lambda}(Y)$ it follows that
\begin{align*}
R_t^{\Lambda}(X) &= R_s^{\Lambda}(X) \cap \LdiF{t} \subseteq \lrsquare{\bigcup_{Y \in {\bf Y}} R_s^{\Lambda}(Y)} \cap \LdiF{t} = \bigcup_{Y \in {\bf Y}} \lrsquare{R_s^{\Lambda}(Y) \cap \LdiF{t}} = \bigcup_{Y \in {\bf Y}} R_t^{\Lambda}(Y).
\end{align*}
\end{example}

We will now briefly present the cocycle condition for the positive convex conjugates $\seq{\beta}$ and $\seq{\alpha}$, which have been proven for the negative conjugates in~\cite{FR12b}. In these results $\beta_{t,s}$ and $\alpha_{t,s}$ correspond with the positive convex conjugates for the stepped risk measures $R_{t,s}$ defined above. Recall from~\cite{FR12b} that a conditional risk measure $R_t$ at time $t$ is called conditionally convex upper continuous (c.c.u.c.) if $R_t^{-1}(D) := \lrcurly{X \in \LdpF{}\;|\; R_t(X) \cap D\neq \emptyset}$ is closed for any conditionally convex closed set $D \in \mathcal{G}(M_t;M_{t,-})$.
\begin{theorem}\label{thm_penalty}
Let $\seq{R}$ be a normalized c.u.c.\ convex risk measure.  Then $\seq{R}$ is multiportfolio time consistent if and only if
\[\beta_t(\Q,w) = \cl\lrparen{\beta_{t,s}(\Q,w) + \EQt{\beta_s(\Q,w_t^s(\Q,w))}{t}}\]
for every $(\Q,w) \in \W_t$ and all times $0 \leq t < s \leq T$.
\end{theorem}

\begin{theorem}\label{thm_cond_penalty}
Let $\seq{R}$ be a normalized c.c.u.c.\ conditionally convex risk measure with dual representation
\[R_t(X) = \bigcap_{(\Q,w) \in \W_t^e} \lrsquare{\lrparen{\EQt{-X}{t} + \Gamma_t(w)} \cap M_t -^. \alpha_t(\Q,w)}\]
for every $X \in \LdpF{}$ where $\W_t^e := \lrcurly{(\Q,w) \in \W_t \; | \; \Q \in \mathcal{M}^e}$.  Then $\seq{R}$ is multiportfolio time consistent if and only if for every $(\Q,w) \in \W_t$ and all times $0 \leq t < s \leq T$
\[\alpha_t(\Q,w) = \cl\lrparen{\alpha_{t,s}(\Q,w) + \EQt{\alpha_s(\Q,w_t^s(\Q,w))}{t}}.\]
\end{theorem}

\section{Supermartingale Property}
\label{sec_supermtg}

In this section we consider a supermartingale-like property for c.u.c.\ convex set-valued risk measures.  This property is akin to that given in~\cite{FP06,BN09} for the scalar case.

Let us introduce the following notation 
\[V_t^{(\Q,w)}(X) := \cl\lrsquare{R_t(X) + \beta_t(\Q,w)}.\]

\begin{theorem}
\label{thm_supermtg}
Let $\seq{R}$ be a normalized c.u.c.\ convex risk measure.  $\seq{R}$ is multiportfolio time consistent if and only if for all times $0 \leq t < s \leq T$ the following supermartingale relation is satisfied: for every $X \in \LdpF{}$ and $(\Q,w) \in \W_t$
\begin{equation}
\label{SMP}
V_t^{(\Q,w)}(X) \subseteq \EQt{V_s^{(\Q,w_t^s(\Q,w))}(X)}{t}.
\end{equation}
Furthermore, the assumption on c.u.c.\ can be weakened on one side of the equivalence:  If $\seq{R}$ is a normalized closed, convex, multiportfolio time consistent risk measure, then \eqref{SMP} is satisfied.
\end{theorem}

Recall from \cite{FR12b} that $\lrcurly{(\Q,w_t^s(\Q,w)) \; | \; (\Q,w) \in \W_t} \subseteq \W_s$ and completely characterizes the dual set $\W_s$ for $t < s$, i.e., for any $(\R,v) \in \W_s$ there exists a $(\Q,w) \in \W_t$ so that for every $X \in \LdpF{}$ it follows that $\lrparen{\EQt{X}{s} + G_s(w_t^s(\Q,w))} \cap M_s = \lrparen{\ERt{X}{s} + G_s(v)} \cap M_s$. As a consequence, the multiportfolio time consistency is equivalent to the supermartingale property of $V_t^{(\Q,w_0^t(\Q,w))}(X)$ holding for all $(\Q,w) \in \W_0$.

Theorem~\ref{thm_supermtg} will be proven with help of the following two lemmas. The proofs of the lemmas can be found in the Appendix.
\begin{lemma}
\label{lemma1}
Under the assumptions of Theorem~\ref{thm_supermtg}, the supermartingale relation of Theorem~\ref{thm_supermtg} holds if and only if the following is satisfied
\begin{align}
    \label{eq_supermtg-1} R_t(X) &\supseteq \bigcup_{Z \in R_s(X)} R_t(-Z)\\
    \label{eq_supermtg-2} R_t(X) &\subseteq \bigcap_{(\Q,w) \in \W_t} \cl \bigcup_{Z \in R_s(X)} \lrsquare{\lrparen{\EQt{Z}{t} + G_t(w)} \cap M_t -^. \beta_{t,s}(\Q,w)}.
\end{align}
\end{lemma}

\begin{lemma}
\label{lemma2}
Under the assumptions of Theorem~\ref{thm_supermtg}, \eqref{eq_supermtg-1}, \eqref{eq_supermtg-2} are equivalent to
\begin{align}
    \label{eq_supermtg-acceptance-1} A_t &\supseteq A_s + A_{t,s}\\
    \label{eq_supermtg-acceptance-2} A_t &\subseteq \bigcap_{(\Q,w) \in \W_t} \lrsquare{A_s + \cl\lrparen{A_{t,s} + G_s^M(w_t^s(\Q,w))}}.
\end{align}
\end{lemma}

\begin{proof}[Proof of Theorem~\ref{thm_supermtg}]
Using Lemmas~\ref{lemma1} and~\ref{lemma2} it remains to show that \eqref{eq_supermtg-acceptance-1} and \eqref{eq_supermtg-acceptance-2} are equivalent to multiportfolio time consistency.
Clearly, multiportfolio time consistency implies \eqref{eq_supermtg-acceptance-1} and \eqref{eq_supermtg-acceptance-2}, see e.g.\ Theorem~\ref{thm_equiv_tc}.  To prove the converse, let $\seq{R}$ satisfy \eqref{eq_supermtg-acceptance-1} and \eqref{eq_supermtg-acceptance-2}.  
The crucial observation is that 
\[
\lrcurly{w_t^s(\Q,w)\;|\; (\Q,w) \in \W_t} = \lrcurly{\Et{Y}{s} \; | \; Y \in \LdqF{+},\;\; \Et{Y}{t} \not\in \prp{M_t}},
\]
which follows from \cite[Lemma 4.5]{FR12}.
Since $M = \R^m \times \{0\}^{d-m}$, $Y \in \LdqF{+}$ implies $\Et{Y}{t} \in \prp{M_t}$ if and only if $Y \in \prp{M_T}$ for any time $t$.  Thus, one obtains
\begin{align*}
A_t &\subseteq \bigcap_{(\Q,w) \in \W_t} \lrsquare{A_s + \cl\lrparen{A_{t,s} + G_s^M(w_t^s(\Q,w))}}\\
&= \bigcap_{\substack{Y \in \LdqF{+}:\\ \Et{Y}{t}\not\in \prp{M_t}}} \lrsquare{A_s + \cl\lrparen{A_{t,s} + G_s^M(\Et{Y}{s})}}\\
&= \bigcap_{Y \in \LdqF{+}} \lrsquare{A_s + \cl\lrparen{A_{t,s} + G_s^M(\Et{Y}{s})}} \subseteq \bigcap_{Y \in \LdqF{+}} \lrsquare{A_s + \cl\lrparen{A_{t,s} + G_T(Y)}}\\
&\subseteq \bigcap_{Y \in \LdqF{+}} \cl\lrsquare{A_s + A_{t,s} + G_T(Y)} \subseteq \bigcap_{Y \in \LdqF{+}} \cl\lrsquare{\cl\lrparen{A_s + A_{t,s}} + G_T(Y)}\\
&= \cl\lrparen{A_s + A_{t,s}} \subseteq \cl(A_t) = A_t.
\end{align*}
Here, the third line follows from $G_s^M(\Et{Y}{s}) = M_s$ if $\Et{Y}{s} \in \prp{M_s}$ and since 
\[A_s + M_s \supseteq \bigcap_{\substack{Y \in \LdqF{+}:\\ \Et{Y}{s} \not\in \prp{M_s}}} \lrsquare{A_s + \cl\lrparen{A_{t,s} + G_s^M(\Et{Y}{s})}}.\]
The last line in the above sequence of equations and inclusions follows from a separation argument between $\cl(A_s + A_{t,s})$ and $\bigcap_{Y \in \LdpF{+}} \cl[\cl(A_s + A_{t,s}) + G_T(Y)]$ since for any $Y \in \LdqF{+}$
\[\cl\lrsquare{\cl\lrparen{A_s + A_{t,s}} + G_T(Y)} = \lrcurly{X \in \LdpF{} \; | \; \E{\trans{Y}X} \geq \inf_{Z \in \cl\lrparen{A_s + A_{t,s}}} \E{\trans{Y}Z}}.\]
The final inclusion is directly from~\eqref{eq_supermtg-acceptance-1}, and the final equality is from $A_t$ closed by assumption of convex upper continuity.

Therefore, $A_t = \cl\lrparen{A_s + A_{t,s}}$, and by \cite[Lemma~B.4]{FR12b} it follows that $A_s + A_{t,s}$ is closed.  Thus $\seq{R}$ is multiportfolio time consistent by \cite[Theorem 3.4]{FR12}.

The last assertion of the theorem holds by noting that the chain of implications from multiportfolio time consistency to \eqref{eq_supermtg-acceptance-1}, \eqref{eq_supermtg-acceptance-2} to \eqref{eq_supermtg-1}, \eqref{eq_supermtg-2} to the supermartingale property does not use c.u.c.\ in addition to closedness.
\end{proof}

\begin{corollary}
\label{cor_supermtg-coherent}
Let $\seq{R}$ be a normalized c.u.c.\ coherent risk measure.  $\seq{R}$ is multiportfolio time consistent if and only if for all times $0 \leq t < s \leq T$
\[V_t^{(\Q,w)}(X) = \cl\lrsquare{R_t(X) + G_t^M(w)} \subseteq \EQt{V_s^{(\Q,w_t^s(\Q,w))}(X)}{t}\] for every $(\Q,w) \in \W_t^{\max}$ and $X \in \LdpF{}$.
Furthermore, if $\seq{R}$ is a normalized closed, coherent, multiportfolio time consistent risk measure, then the supermartingale property is satisfied.
\end{corollary}
\begin{proof}
This follows from Theorem~\ref{thm_supermtg} by noting that, as a consequence of coherence, \[\beta_t(\Q,w) = \begin{cases}G_t^M(w) &\text{if }(\Q,w) \in \W_t^{\max}\\ \emptyset &\text{else}\end{cases}.\]  
If $(\Q,w) \not\in \W_t^{\max}$, then $V_t^{(\Q,w)}(X) = \emptyset$ (recalling that $R_t(X) + \emptyset = \emptyset$) and the supermartingale property is trivially satisfied. If $(\Q,w) \in \W_t^{\max}$, then $V_t^{(\Q,w)}(X) 
= \cl\lrsquare{R_t(X) + G_t^M(w)}$. 
\end{proof}

\begin{example}
\label{ex:systemic-smtg}
\emph{Aggregation-based risk measure}: Consider again the aggregation-based risk measure constructed in Example~\ref{ex:systemic-defn}. As this is multiportfolio time consistent (see Example~\ref{ex:systemic-mptc}), by Corollary~\ref{cor_supermtg-coherent} and the dual representation of the aggregation-based risk measure provided in Example~\ref{ex:systemic-dual}, 
\begin{align*}
V_t^{(\Q,w)}(X) &:= \cl\lrsquare{R_t^{\Lambda}(X) + G_t(w)}\\
&= \cl\lrcurly{u_R + u_G \; | \; u_R,u_G \in \LdiF{t}, \; \E{\trans{w}u_G} \geq 0, \; \trans{v_k}u_R \geq -\trans{v_k}X \; \forall k = 0,1,...,2^d-1}
\end{align*}
is a set-valued supermartingale for any $(\Q,w) \in \W_t^{\Lambda}$, i.e., $(\Q,w) \in \W_t$ such that $w_t^T(\Q,w) \in \LdpK{1}{}{\plus{\Lambda^{-1}[\R_+]}}$.  In this case, in particular, letting $w = v_k$ for some index $k$ and $\Q \in \mathcal{M}^e$, we have $(\Q,v_k) \in \W_t^{\Lambda}$ and $V_t^{(\Q,v_k)} = \lrcurly{u \in \LdiF{t} \; | \; \E{\trans{v_k}u} \geq \E{\rho_t^{WC}(\trans{v_k}X)}}$ is a $\Q$-supermartingale.  
Note the connection with the scalar result, from \cite[Corollary 4.12]{FP06}, that $\rho_t^{WC}(\trans{v_k}X)$ is a $\Q$-supermartingale for any $\Q \in \mathcal{M}^e$ and any $k$ due to the time consistency of the worst case risk measure.  We will elaborate on this connection in Example~\ref{ex:systemic-smtg-cc} below.
\end{example}

We will now identify those dual variables that make $V_t$ a martingale as the ``worst-case''  dual variables in the dual representation.
Compare to Proposition 1.21 in~\cite{AP10} for the scalar case. 
\begin{corollary}
\label{cor_mtg}
Let $\seq{R}$ be a normalized c.u.c., convex, multiportfolio time consistent risk measure and fix $X \in \LdpF{}$.  $\lrparen{V^{(\Q,w_0^t(\Q,w))}_t(X)}_{t = 0}^T$ is a $\Q-$martingale, i.e.
\[V_t^{(\Q,w_0^t(\Q,w))}(X) = \EQt{V_s^{(\Q,w_0^s(\Q,w))}(X)}{t} \quad \forall 0 \leq t < s \leq T,\]
for any $(\Q,w) \in \W_0$ that satisfy the two conditions $\beta_0(\Q,w) \neq \emptyset$ and
\[\cl\lrsquare{R_0(X) + G_0^M(w)} = \lrparen{\EQ{-X} + G_0(w)} \cap M_0 -^. \beta_0(\Q,w).\]
Additionally, this choice of $(\Q,w)$ is a ``worst-case'' pair of dual variables for $X$ at any time $t$, i.e.,
\[\cl\lrsquare{R_t(X) + G_t^M(w_0^t(\Q,w))} = \lrparen{\EQt{-X}{t} + G_t(w_0^t(\Q,w))} \cap M_t -^. \beta_t(\Q,w_0^t(\Q,w)).\]
If $M = \R^d$ then, conversely, if $\lrparen{V_t^{(\Q,w_0^t(\Q,w))}(X)}_{t = 0}^T$ is a $\Q-$martingale for some $(\Q,w) \in \W_0$ with $\beta_0(\Q,w) \neq \emptyset$, then $(\Q,w)$ is a ``worst-case'' pair of dual variables for $X$ for any time $t$.
\end{corollary}
\begin{proof}
See appendix, Section~\ref{sec_proof_cor_mtg}.
\end{proof}

\begin{remark}
The supermartingale relation can be given with the negative conjugates $\seq{-\bar\beta}$ (see~\cite{FR12,FR12b} or Appendix~\ref{appendix_prelim}) though it requires additional considerations due to the fact that $\cl\lrsquare{R_t(X) + G_t^M(w)} -^. (-\bar\beta_t(\Q,w)) \neq \cl\lrsquare{R_t(X) + \beta_t(\Q,w)}$ when $-\bar\beta_t(\Q,w) = M_t$ (or equivalently when $\beta_t(\Q,w) = \emptyset$) and $\cl\lrsquare{R_t(X) + G_t^M(w)} = M_t$.
\end{remark}

\begin{example}\label{ex_entropic}
\emph{Restrictive entropic risk measure}: Consider the full space of eligible assets for all times $t$, i.e., $M_t = \LdpF{t}$.  The restrictive entropic risk measure with parameter $\lambda \in \R^d_{++}$
\[R_t^{ent}(X) = \lrcurly{u \in \LdpF{t} \; | \; \E{1 - \exp\lrparen{-\diag{\lambda}(X + u)}} \in \LdpF{t,+}}\] 
is a normalized c.u.c.\ convex risk measure that is multiportfolio time consistent.  For details see~\cite[Example 3.4 and Section 6.2]{FR12b} and \cite{AHR13}.  
By Theorem~\ref{thm_supermtg} we obtain that, with conditional relative entropy $\hat H_t(\Q | \P) = \EQt{\log\lrparen{\dQdP}}{t}$,
\[V_t^{(\Q,w)}(X) := \cl\lrsquare{R_t^{ent}(X) + \lrparen{\diag{\frac{1}{\lambda}} \hat H_t(\Q | \P) + G_t(w)}}\]
is a set-valued supermartingale for any $(\Q,w) \in \W_t$.
\end{example}
\begin{example}\label{ex_avar}
\emph{Composed average value at risk}: Consider a discrete time setting with the full space of eligible assets for all times $t$, i.e., $M_t = \LdpF{t}$.  The average value at risk $AV@R_t(X)$ (for any parameter $\lambda^t \in \LdiF{t,++}$ bounded away from $0$) defines a normalized c.u.c.\ coherent dynamic risk measure which is not multiportfolio time consistent.  However, the composition of the average value at risk $\widetilde{AV@R}_t(X) := AV@R_t\lrparen{-\widetilde{AV@R}_{t+1}(X)}$ is multiportfolio time consistent.  For details see~\cite[Section 5.2]{FR12}, \cite[Example 5.5 and Section 6.1]{FR12b}, and \cite{HRY12}.  
By Corollary~\ref{cor_supermtg-coherent},
\[V_t^{(\Q,w)}(X) := \cl\lrsquare{\widetilde{AV@R}_t(X) + G_t(w)}\]
is a set-valued supermartingale for any $(\Q,w) \in \widetilde{\W}_t$, where
\begin{align*}
\widetilde{\W}_t &:= \Big\{(\Q,w) \in \W_t \; | \; \forall i \in \{1,...,d\} \forall s \in \{t,...,T-1\}:\\
&\quad\quad\quad\quad \rcurly{\P\lrparen{w_i = 0 \text{ or } \bar\xi_{s,s+1}(\Q_i) \leq \frac{1}{\lambda^t_i}} = 1}.
\end{align*}
\end{example}

\section{Conditional Supermartingale Property}
\label{sec_supermtg-cc}

Now we extend the results of the previous section to the conditional penalty function $\alpha$.  That is, we present a supermartingale property for the set-valued stochastic process
\[ \V_t^{(\Q,w)}(X) := \cl\lrsquare{R_t(X) + \alpha_t(\Q,w)} \]
for c.u.c.\ conditionally convex dynamic risk measures $\seq{R}$.

\begin{corollary}
\label{cor_supermtg-cc}
Let $\seq{R}$ be a normalized c.u.c.\ conditionally convex risk measure with dual representation with equivalent probability measures only, i.e., 
\[R_t(X) = \bigcap_{(\Q,w) \in \W_t^e} \lrsquare{\lrparen{\EQt{-X}{t} + \Gamma_t(w)} \cap M_t -^. \alpha_t(\Q,w)}.\]
Then, $\seq{R}$ is multiportfolio time consistent if and only if for all times $0 \leq t < s \leq T$
\begin{equation}\label{condSMP}
\V_t^{(\Q,w)}(X) \subseteq \cl\EQt{\V_s^{(\Q,w_t^s(\Q,w))}(X)}{t}
\end{equation}
for every $(\Q,w) \in \W_t^e$ and $X \in \LdpF{}$. 
Furthermore, if $\seq{R}$ is a normalized closed, conditionally convex, multiportfolio time consistent risk measure, then the supermartingale property \eqref{condSMP} is satisfied.
\end{corollary}

\begin{corollary}
\label{cor_supermtg-cc-coherent}
Let $\seq{R}$ be a normalized c.u.c.\ conditionally coherent risk measure with dual representation with equivalent probability measures only, i.e., 
\[R_t(X) = \bigcap_{(\Q,w) \in \W_t^{e,\max}} \lrparen{\EQt{-X}{t} + \Gamma_t(w)} \cap M_t\]
where $\W_t^{e,\max} := \W_t^{\max} \cap \W_t^e$.
$\seq{R}$ is multiportfolio time consistent if and only if for all times $0 \leq t < s \leq T$ 
\[\V_t^{(\Q,w)}(X) = \cl\lrsquare{R_t(X) + \Gamma_t^M(w)} \subseteq \EQt{\V_s^{(\Q,w_t^s(\Q,w))}(X)}{t}\] for every $(\Q,w) \in \W_t^{e,\max}$. 
Furthermore, if $\seq{R}$ is a normalized closed, conditionally coherent, multiportfolio time consistent risk measure, then the supermartingale property is satisfied.
\end{corollary}
\begin{proof}
This follows from Corollary~\ref{cor_supermtg-cc} since
\[\alpha_t(\Q,w) = \begin{cases}\Gamma_t^M(w) &\text{if } (\Q,w) \in \W_t^{\max} \\ \emptyset & \text{else}\end{cases}\] 
and thus $\V_t^{(\Q,w)}(X) 
= \cl\lrsquare{R_t(X) + \Gamma_t^M(w)}$ for any $(\Q,w) \in \W_t^{\max}$.  
\end{proof}

\begin{example}
\label{ex:systemic-smtg-cc}
\emph{Aggregation-based risk measure}: Consider again the aggregation-based risk measure constructed in Example~\ref{ex:systemic-defn}. As this is multiportfolio time consistent (see Example~\ref{ex:systemic-mptc}), by Corollary~\ref{cor_supermtg-coherent} and the dual representation of the aggregation-based risk measure provided in Example~\ref{ex:systemic-dual}, 
\begin{align*}
\V_t^{(\Q,w)}(X) &:= \cl\lrsquare{R_t^{\Lambda}(X) + \Gamma_t(w)}\\
&= \cl\lrcurly{u_R + u_{\Gamma} \; | \; u_R,u_{\Gamma} \in \LdiF{t}, \; \trans{w}u_{\Gamma} \geq 0, \; \trans{v_k}u_R \geq -\trans{v_k}X \; \forall k = 0,1,...,2^d-1}
\end{align*}
is a set-valued supermartingale for any $(\Q,w) \in \W_t^{\Lambda}$, i.e., $(\Q,w) \in \W_t$ such that $w_t^T(\Q,w) \in \LdpK{1}{}{\plus{\Lambda^{-1}[\R_+]}}$.  In this case, in particular, letting $w = v_k$ for some index $k$ and $\Q \in \mathcal{M}^e$, we have $(\Q,v_k) \in \W_t^{\Lambda}$ and $\V_t^{(\Q,v_k)} = \lrcurly{u \in \LdiF{t} \; | \; \trans{v_k}u \geq \rho_t^{WC}(\trans{v_k}X)}$ is a $\Q$-supermartingale.  This provides a direct comparison between the supermartingale property for the (time consistent) scalar worst case risk measure and the set-valued aggregation-based risk measure considered here.  That is, the set-valued supermartingale property implies the supermartingale property of the underlying scalar risk measure, but the converse is not necessarily true, i.e., more information is provided by the set-valued supermartinale property than just the individual scalar property as it is over every $(\Q,w) \in \W_t^{\Lambda}$ without the restriction of $w \in \{v_k \; | \; k = 0,1,...,2^d-1\}$.
\end{example}

Again, we can characterize the dual variables under which $\V_t$ is a martingale as the ``worst-case'' dual variables.
\begin{corollary}
\label{cor_mtg-cc}
Let $\seq{R}$ be a normalized c.u.c., conditionally convex, multiportfolio time consistent risk measure with dual representation with equivalent probability measures only, i.e.,
\[R_t(X) = \bigcap_{(\Q,w) \in \W_t^e} \lrsquare{\lrparen{\EQt{-X}{t} + \Gamma_t(w)} \cap M_t -^. \alpha_t(\Q,w)}.\]
Fix some $X \in \LdpF{}$.  For any $(\Q,w) \in \W_0^e$ satisfying $\alpha_0(\Q,w) \neq \emptyset$ and
\[\cl\lrsquare{R_0(X) + \Gamma_0^M(w)} = \lrparen{\EQ{-X} + \Gamma_0(w)} \cap M_0 -^. \alpha_0(\Q,w),\]
the stochastic process $\lrparen{V^{(\Q,w_0^t(\Q,w))}_t(X)}_{t = 0}^T$ is a $\Q-$martingale, i.e.
\[\V_t^{(\Q,w_0^t(\Q,w))}(X) = \cl\EQt{\V_s^{(\Q,w_0^s(\Q,w))}(X)}{t} \quad \forall 0 \leq t < s \leq T.\]
Additionally, this choice of $(\Q,w)$ is a ``worst-case'' pair of dual variables for $X$ at any time $t$, i.e.,
\[\cl\lrsquare{R_t(X) + \Gamma_t^M(w_0^t(\Q,w))} = \lrparen{\EQt{-X}{t} + \Gamma_t(w_0^t(\Q,w))} \cap M_t -^. \alpha_t(\Q,w_0^t(\Q,w)).\]
If $M = \R^d$ then, conversely, if $\lrparen{\V_t^{(\Q,w_0^t(\Q,w))}(X)}_{t = 0}^T$ is a $\Q-$martingale for some $(\Q,w) \in \W_0^e$ with $\alpha_0(\Q,w) \neq \emptyset$ then $(\Q,w)$ is a ``worst-case'' pair of dual variables for $X$ for any time $t$.
\end{corollary}

\begin{example}\label{ex_shp}
\emph{Superhedging}: Consider a discrete time setting with the full space of eligible assets for all times $t$, i.e., $M_t = \LdpF{t}$.  Define a market with convex transaction costs described by convex solvency regions $\seq{K}$ with \[\interior\recc{K_t[\omega]} \supseteq \R^d_+\backslash\{0\}\] almost surely, where $\recc C$ denotes the recession cone of a convex set $C\subseteq \R^d$.  The set of superhedging portfolios $SHP_t(X)$ at time $t$ denotes those eligible portfolios that can be traded from time $t$ to the terminal time $T$ to outperform the input portfolio $X \in \LdpF{}$.  From this formulation we can immediately define a closed, conditionally convex multivariate risk measure $R_t(X) := SHP_t(-X)$ which is multiportfolio time consistent.  For details see~\cite[Example 5.4]{FR12b}.  From the constraint on the interior of the solvency regions, the penalty function is only defined on the set of equivalent probability measures.
Though $SHP_t$ is not normalized in general, we may still apply the results of this section as the summation of acceptance sets and penalty function representations holds via a composition backwards in time (see \cite[Section~5]{FR12b}) and thus all proofs follow accordingly. By Corollary~\ref{cor_supermtg-cc} we obtain that
\[\V_t^{(\Q,w)}(X) := \cl\lrsquare{SHP_t(-X) + \alpha_t^{SHP}(\Q,w)}\]
defines a set-valued supermartingale for any $(\Q,w) \in \W_t^e$, where the penalty function is given by
\[\alpha_t^{SHP}(\Q,w) := \sum_{s = t}^T \lrcurly{u \in \LdpF{t} \; | \; \esssup_{k_s \in \LdpK{p}{s}{K_s}} -\trans{w}\EQt{k_s}{t} \leq \trans{w}u }.\]
In the special case that the market has proportional transaction costs only, i.e., the market is defined by a sequence of solvency cones $\seq{K}$, then the set of superhedging portfolios defines a normalized, closed, conditionally coherent risk measure which is multiportfolio time consistent.  For details see~\cite[Section 5.1]{FR12}, \cite[Example 4.7]{FR12b}, and \cite{LR11}.  Then, by Corollary~\ref{cor_supermtg-cc-coherent} we obtain that
\[\V_t^{(\Q,w)}(X) := \cl\lrsquare{SHP_t(-X) + \Gamma_t(w)}\]
defines a set-valued supermartingale for any $(\Q,w) \in \W_{\{t,...,T\}}$, where
\[\W_{\{t,...,T\}} := \lrcurly{(\Q,w) \in \W_t^e \; | \; w_t^s(\Q,w) \in \LdpK{q}{s}{\plus{K_s}} \forall s \in \{t,...,T\}}.\]
Similarly we could apply Theorem~\ref{thm_supermtg} (Corollary~\ref{cor_supermtg-coherent} in the case of proportional transaction costs), which yields that $V_t^{(\Q,w)}(X)$ is a set-valued supermartingale as well.
\end{example}

\appendix
\section{Dual representations for (conditional) scalarizations of set-valued risk measures}\label{appendix_scalar}
The proofs of the main results of Sections~\ref{sec_supermtg} and \ref{sec_supermtg-cc} will use the following results on dual representations for (conditional) scalarizations of set-valued risk measures, which are of independent interest.

\subsection{Set-valued scalarization}
\begin{proposition}
\label{prop_scalar}
Let $R_t$ be a c.u.c.\ and convex risk measure. Consider $w \in \plus{\recc{R_t(0)}}\backslash \prp{M_t}$, i.e. $\E{\trans{w}u} \geq 0$ for every $u \in \recc{R_t(0)} := \{m \in M_t \; | \; R_t(0) + m \subseteq R_t(0)\}$ and there exists some $m \in M_t$ where $\E{\trans{w}m} \neq 0$. Then, for every $X \in \LdpF{}$  the following holds
\begin{equation}
\label{dualsc}
\rho_t(X) :=\inf_{u \in R_t(X)} \E{\trans{w}u} = \sup_{(\Q,m_{\perp}) \in \W_t(w)} \inf_{Y \in A_t} \E{\transp{w + m_{\perp}}\EQt{Y - X}{t}},
\end{equation}
where $\W_t(w) := \lrcurly{(\Q,m_{\perp}) \in \mathcal{M} \times \prp{M_t} \; | \; (\Q,w + m_{\perp}) \in \W_t}$.
\end{proposition}
\begin{proof}
Clearly, $\rho_t(X) = \inf_{u \in \cl\lrparen{R_t(X) + G_t^M(w)}} \E{\trans{w}u}$.  It holds by a separation argument that
\begin{align*}
A_t^w &:= \lrcurly{X \in \LdpF{} \; | \; \rho_t(X) \leq 0}\\ 
&= \lrcurly{X \in \LdpF{} \; | \; 0 \in \cl\lrparen{R_t(X) + G_t^M(w)}} = \cl_M\lrparen{A_t + G_t^M(w)}.
\end{align*}  
In fact, $A_t^w = \cl_M\lrparen{A_t + G_t^M(w+m_{\perp})}$ for any $m_{\perp} \in \prp{M_t}$.  See \cite[Definition 2.15]{HHR10} for a discussion of the directional closure $\cl_M$.

The biconjugate is given by
$\rho_t^{**}(X) = \sup_{Z \in \LdqF{}} \lrparen{\E{\trans{Z}X} - \rho_t^*(Z)}$,
where the convex conjugate is defined by
$\rho_t^*(Z) = \sup_{X \in \LdpF{}} \lrparen{\E{\trans{Z}X} - \rho_t(X)}$.
Let us calculate the effective domain of $\rho_t^*$.
    Consider first $Z \not\in \LdqF{-}$.  Let $\hat Y \in A_t$ and $Y \in \LdpF{+}$ such that $\E{\trans{Z}Y} > 0$.  In particular, $\hat Y + \lambda Y \in A_t$ for any $\lambda > 0$.  Therefore,
    \begin{align*}
        \rho_t^*(Z) &= \sup_{X \in \LdpF{}} \lrparen{\E{\trans{Z}X} - \rho_t(X)} \geq \sup_{\lambda > 0} \lrparen{\E{\trans{Z}(\hat Y + \lambda Y)} - \rho_t(\hat Y + \lambda Y)}\\
        &\geq \sup_{\lambda > 0} \E{\trans{Z}(\hat Y + \lambda Y)} = \E{\trans{Z}\hat Y} + \sup_{\lambda > 0} \lambda \E{\trans{Z}Y} = \infty.
    \end{align*}
Now consider $\Et{Z}{t} \not\in -w + \prp{M_t}$.  Let $m \in G_t^M(\Et{Z}{t}) \cap G_t^M(w)$ such that $\E{\trans{Z}m} + \E{\trans{w}m} > 0$.  Therefore,
    \begin{align*}
    \rho_t^*(Z) &= \sup_{X \in \LdpF{}} \lrparen{\E{\trans{Z}X} - \rho_t(X)} \geq \sup_{\lambda > 0} \lrparen{\E{\trans{Z}(\lambda m)} - \rho_t(\lambda m)} \\
    &= -\rho_t(0) + \sup_{\lambda > 0} \lambda \lrparen{\E{\trans{Z}m} + \E{\trans{w}m}} = \infty.
    \end{align*}
This implies that $\rho_t^*(Z) \neq \infty$ only if $Z = -w_t^T(\Q,w+m_{\perp})$ for some $(\Q,m_{\perp}) \in \W_t(w)$.

It remains to show that $\rho_t^*(-w_t^T(\Q,w+m_{\perp})) = \sup_{Y \in A_t}\E{-(\trans{w+m_{\perp})}\EQt{Y}{t}}$ for any $(\Q,m_{\perp}) \in \W_t(w)$. We will prove this by two inequalities.
    Let $(\Q,m_{\perp}) \in \W_t(w)$.
    \begin{align*}
   & \rho_t^*(-w_t^T(\Q,w+m_{\perp})) = \sup_{X \in \LdpF{}} \lrparen{\E{-\transp{w+m_{\perp}}\EQt{X}{t}} - \rho_t(X)}
    \\
    & \geq \sup_{Y \in A_t^w} \lrparen{\E{-\transp{w+m_{\perp}}\EQt{Y}{t}} - \rho_t(Y)}
    \\&\geq \sup_{Y \in A_t^w} \E{-\transp{w+m_{\perp}}\EQt{Y}{t}} 
    \geq \sup_{Y \in A_t} \E{-\transp{w+m_{\perp}}\EQt{Y}{t}}.
    \end{align*}
    For the reverse inequality, denote for those $X \in \LdpF{} $ with $R_t(X) \neq \emptyset$ by $u_X \in M_t$ the random variable satisfying $\E{\trans{w}u_X} = \rho_t(X)$. Then, $X + u_X \in A_t^w$ and
    \begin{align}
       \nonumber& \rho_t^*(-w_t^T(\Q,w+m_{\perp})) = \sup_{X \in \LdpF{}} \lrparen{\E{-\transp{w+m_{\perp}}\EQt{X}{t}} - \rho_t(X)} 
        \\
       \nonumber &= \sup_{\substack{X \in \LdpF{}:\\ R_t(X) \neq \emptyset}} \lrparen{\E{-\transp{w+m_{\perp}}\EQt{X}{t}} - \rho_t(X)}\\
      \nonumber  &= \sup_{\substack{X \in \LdpF{}:\\ R_t(X) \neq \emptyset}} \E{-\transp{w+m_{\perp}}\EQt{X+u_X}{t}} \leq \sup_{Y \in A_t^w} \E{-\transp{w+m_{\perp}}\EQt{Y}{t}}\\
        \label{eqC1B1}&\leq \sup_{Y \in \cl\lrparen{A_t + G_t^M(w+m_{\perp})}} \E{-\transp{w+m_{\perp}}\EQt{Y}{t}} = \sup_{Y \in A_t} \E{-\transp{w+m_{\perp}}\EQt{Y}{t}}.
    \end{align}
Therefore we obtain that
\[\rho_t^{**}(X) = \sup_{(\Q,m_{\perp}) \in \W_t(w)} \inf_{Y \in A_t} \E{\transp{w+m_{\perp}}\EQt{Y-X}{t}}.\]
As $\rho_t$ is proper (by Corollary~\ref{rem_proper}), convex by convexity of $R_t$, and lower semicontinuous (by Proposition~\ref{prop_lsc} since $\plus{\recc{R_t(0)}}\subseteq \plus{M_{t,+}} $), the Fenchel-Moreau Theorem yields the assertion.
\end{proof}

\begin{corollary}
\label{cor_scalar_stepped}
Let $R_{t,s}$ be a c.u.c.\ and convex stepped risk measure. Consider $w \in \plus{\recc{R_t(0)}}\backslash \prp{M_t}$. Then, for every $X \in M_s$  the following holds
\begin{equation}
\label{dualscst}
\inf_{u \in R_t(X)} \E{\trans{w}u} = \sup_{(\Q,m_{\perp}) \in \W_{t,s}(w)} \inf_{Y \in A_{t,s}} \E{\transp{w + m_{\perp}}\EQt{Y - X}{t}},
\end{equation}
where 
\begin{align*}
\W_{t,s}(w) &:= \lrcurly{(\Q,m_{\perp}) \in \mathcal{M} \times \prp{M_t} \; | \; (\Q,w + m_{\perp}) \in \W_{t,s}} \supseteq \W_t(w)\\
\W_{t,s} &= \lrcurly{(\Q,v) \in \mathcal{M} \times \lrparen{\plus{M_{t,+}} \backslash \prp{M_t}} \; | \; w_t^s(\Q,v) \in \plus{M_{s,+}}, \Q = \P|_{\Ft{t}}}.
\end{align*}
\end{corollary}
\begin{proof}
This follows similarly to Proposition~\ref{prop_scalar}.
\end{proof}

We will now give conditions that ensure the lower semicontinuity and properness of the scalarizations~\eqref{dualsc} and~\eqref{dualscst}.
The following proposition is a modification of \cite[Proposition 3.29]{HS12-uc} by noting that the only open sets needed in that proof are those with a closed and convex complement. However, for the convenience of the reader, a proof is provided here as well.

\begin{proposition}
\label{prop_lsc}
If $R_t$ is c.u.c., the scalarization $\rho_t(X) := \inf_{u \in R_t(X)} \E{\trans{w}u}$ is lower semicontinuous for any $w \in \plus{M_{t,+}} \backslash \prp{M_t}$.
\end{proposition}
\begin{proof}
By definition, $R_t$ is c.u.c.\ if and only if $R_t^+(V) := \lrcurly{X \in \LdpF{} \; | \; R_t(X) \subseteq V}$ is open for any $V \subseteq M_t$ such that the complement $V^c$ is closed and convex.  For any $\epsilon > 0$, $V_\epsilon = \lrcurly{m \in M_t \; | \; \E{\trans{w}m} > -\epsilon}$ is an open neighborhood of $0$.  Thus, $R_t(X_0) + V_\epsilon$ is open for any $X_0$ and the complement is convex via:
\begin{align*}
    [R_t(X_0) + V_\epsilon]^c &= \lrcurly{m \in M_t \; | \; \E{\trans{w}(m-u)} \leq -\epsilon \; \forall u \in R_t(X_0)}\\
    &= \lrcurly{m \in M_t \; | \; \E{\trans{w}m} \leq -\epsilon + \inf_{u \in R_t(X)} \E{\trans{w}u}}.
\end{align*}
Therefore $R_t^+(R_t(X_0) + V_\epsilon)$ is open, and trivially is a neighborhood of $X_0$.  Thus, $\rho_t(X) \geq \rho_t(X_0) - \epsilon$ for any $X \in R_t^+(R_t(X_0) + V_\epsilon)$, which implies lower semicontinuity at $X_0$.  Since this is true for any $X_0$, the result is proven.
\end{proof}

\begin{proposition}
\label{prop_proper}
Let $R_t$ be a closed and convex risk measure.  The scalarization $\rho_t$ is proper ($\rho_t(X) > -\infty$ for every $X \in \LdpF{}$ and $\rho_t(X) < \infty$ for some $X \in \LdpF{}$) if and only if $w \in \bigcap_{\substack{X \in \LdpF{}\\ R_t(X) \neq \emptyset}} \plus{\recc{R_t(X)}}$.
\end{proposition}
\begin{proof}
Clearly, $\rho_t(0) < \infty$ for any $w \in \plus{M_{t,+}}$ since $R_t(0) \neq \emptyset$.  And for any $X \in \LdpF{}$ with $R_t(X) \neq \emptyset$, we know $\rho_t(X) > -\infty$ if and only if $w \in \plus{\recc{R_t(X)}}$ by the definition of the recession cone. For  $X \in \LdpF{}$ with $R_t(X)= \emptyset$, $\rho_t(X) > -\infty$ is trivially true.
\end{proof}

\begin{corollary}
\label{cor_proper}
If $R_t$ is a closed convex risk measure and $X \in \LdpF{}$ with $R_t(X) \neq \emptyset$. Then, $\recc{R_t(X)} = \bigcap_{(\Q,w) \in \W_t^t} G_t^M(w)$.
\end{corollary}
\begin{proof}
$r \in \recc{R_t(X)}$ if and only if $r \in \bigcap_{\lambda > 0} \lambda(R_t(X) - u)$ for any $u \in R_t(X)$ if and only if $u + \frac{1}{\lambda}r \in
R_t(X)$ for any $\lambda > 0$ and any $u \in R_t(X)$.  Using \eqref{Eq:dual} and noting that one can replace $\W_t$ with  $\W_t^t$, this is equivalent to $\E{\trans{w}(u + \frac{1}{\lambda}r)} \geq \inf_{Y \in A_t}\E{\trans{w}\EQt{Y-X}{t}}$ for every $(\Q,w) \in \W_t^t$, every $\lambda > 0$ and every $u \in R_t(X)$.  This is true if and only if for every $(\Q,w) \in \W_t^t$ and $u \in R_t(X)$
\[\E{\trans{w}r} \geq \sup_{\lambda > 0} \lambda \lrsquare{\inf_{Y \in A_t}\E{\trans{w}\EQt{Y-X}{t}} - \E{\trans{w}u}} = 0,\] 
where the last equality holds since $u \in R_t(X)$.  
This yields the assertion.
\end{proof}

\begin{corollary}
\label{rem_proper}
If $R_t$ is a closed convex risk measure, then $\rho_t$ is proper ($\rho_t(X) > -\infty$ for every $X \in \LdpF{}$ and $\rho_t(X) < \infty$ for some $X \in \LdpF{}$) if and only if $w \in \plus{\recc{R_t(0)}} = \plusp{\bigcap_{(\Q,w) \in \W_t^t} G_t^M(w)}$. 
\end{corollary}
\begin{proof}
This follows from Proposition~\ref{prop_proper} and Corollary~\ref{cor_proper}.
\end{proof}

\subsection{Conditional scalarization}\label{sec_condscal}
\begin{proposition}
\label{prop_scalar-cc}
Let $R_t$ be a c.u.c.\ and conditionally convex risk measure. Let $w \in \plus{\recc{R_t(0)}} \backslash \prp{M_t}$. Then, for every $X \in \LdpF{}$
\begin{equation*}
\hat \rho_t(X) := \essinf_{u \in R_t(X)} \trans{w}u = \esssup_{(\Q,m_{\perp}) \in \W_t(w)} \essinf_{Y \in A_t} \transp{w + m_{\perp}}\EQt{Y - X}{t},
\end{equation*}
where $\W_t(w) := \lrcurly{(\Q,m_{\perp}) \in \mathcal{M} \times \prp{M_t} \; | \; (\Q,w + m_{\perp}) \in \W_t}$.
\end{proposition}
\begin{proof}
Denote the acceptance set by 
$\hat A_t^w := \lrcurly{X \in \LdpF{} \; | \; \hat \rho_t(X) \leq 0}$.
We will show 
\begin{equation}\label{eqC1}
\hat \rho_t(X) \geq \esssup_{(\Q,m_{\perp}) \in \W_t(w)} \essinf_{Y \in \hat A_t^w} \transp{w+m_{\perp}}\EQt{Y-X}{t}.
\end{equation}
The inequality is trivially satisfied if $R_t(X) = \emptyset$ since then $\hat \rho_t(X) = \infty$ almost surely. Thus, assume $R_t(X) \neq \emptyset$.  Since $\E{\hat \rho_t(X)} = \inf_{u \in R_t(X)} \E{\trans{w}u}$, this implies $\hat \rho_t(X) \in \LoF{t}$. Thus, there exists some $u_X \in M_t$ such that $\trans{w}u_X = \hat \rho_t(X)$.  By translativity, this implies that $X+u_X \in \hat A_t^w$, which in turn implies \[\transp{w+m_{\perp}}\EQt{X+u_X}{t}\geq \essinf_{Y \in \hat A_t^w} \transp{w+m_{\perp}}\EQt{Y}{t} \] for every $(\Q,m_{\perp}) \in \W_t(w)$. Subtracting $\transp{w+m_{\perp}}\EQt{X}{t}$ on both sides of the inequality  and then taking the essential supremum over $(\Q,m_{\perp}) \in \W_t(w)$ yields \eqref{eqC1}.
Now we want to show that 
equality holds in \eqref{eqC1}. In combination with \eqref{eqC1}, we will do this by showing that 
\[
\E{\hat \rho_t(X)} \leq \E{\esssup_{(\Q,m_{\perp}) \in \W_t(w)} \essinf_{Y \in \hat A_t^w} \transp{w+m_{\perp}}\EQt{Y-X}{t}}
\]
for every $X \in \LdpF{}$.  By decomposability and Proposition~\ref{prop_scalar} in conjunction with \eqref{eqC1B1} one obtains
\[
\E{\hat \rho_t(X)} = \sup_{(\Q,m_{\perp}) \in \W_t(w)} \inf_{Y \in A_t^w} \E{\transp{w+m_{\perp}}\EQt{Y-X}{t}}.
\]  
Since $\hat A_t^w \subseteq A_t^w$ and by decomposability of the necessary sets, the desired result is immediate.

It remains to show that we can replace $\hat A_t^w$ with $A_t$.  First note that $\hat A_t^w \supseteq A_t$, implying $\hat \rho_t(X) \leq \esssup_{(\Q,m_{\perp}) \in \W_t(w)} \essinf_{Y \in A_t} \transp{w+m_{\perp}}\EQt{Y-X}{t}$.  Let $(u_n)_{n \in \N} \subseteq R_t(X)$ so that $\trans{w}u_n \searrow \hat \rho_t(X)$ w.r.t.\ almost sure convergence.  The existence of such a sequence follows from~\cite[Theorem A.33(b)]{FS11} because decomposability of $R_t(X)$ implies the assumption of that theorem. Hence,
\begin{align*}
\esssup_{(\Q,m_{\perp}) \in \W_t(w)} &\essinf_{Y \in A_t} \transp{w+m_{\perp}}\EQt{Y-X}{t} \\
&\leq \esssup_{(\Q,m_{\perp}) \in \W_t(w)} \overline\lim_{n \to \infty} \transp{w+m_{\perp}}\EQt{(X+u_n)-X}{t}\\
&= \esssup_{(\Q,m_{\perp}) \in \W_t(w)} \overline\lim_{n \to \infty} \trans{w}u_n = \hat \rho_t(X),
\end{align*}
where the limit, $\overline\lim$, is taken in the almost sure sense.
\end{proof}

\begin{corollary}
\label{cor_scalar_stepped-cc}
Let $R_{t,s}$ be a c.u.c.\ and conditionally convex risk measure. Let $w \in \plus{\recc{R_t(0)}} \backslash \prp{M_t}$. Then, for every $X \in M_s$
\begin{equation*}
\essinf_{u \in R_t(X)} \trans{w}u = \esssup_{(\Q,m_{\perp}) \in \W_{t,s}(w)} \essinf_{Y \in A_{t,s}} \transp{w + m_{\perp}}\EQt{Y - X}{t},
\end{equation*}
where $\W_{t,s}(w) := \lrcurly{(\Q,m_{\perp}) \in \mathcal{M} \times \prp{M_t} \; | \; (\Q,w + m_{\perp}) \in \W_{t,s}}$.
\end{corollary}
\begin{proof}
This follows similarly to Proposition~\ref{prop_scalar-cc}.
\end{proof}

Recall the notation $\W_t^s =\lrcurly{(\Q,w) \in \W_t \; | \; \beta_s(\Q,w_t^s(\Q,w)) \neq \emptyset}$. The next proposition shows that this set coincides with $\widehat \W_t^s:=\lrcurly{(\Q,w) \in \W_t \; | \; \alpha_s(\Q,w_t^s(\Q,w)) \neq \emptyset}$.

\begin{proposition}\label{prop_cond_Wts}
$\W_t^s = \widehat \W_t^s$ for any times $0 \leq t \leq s \leq T$.
\end{proposition}

\begin{proof}
Let $(\Q,w) \in \W_t^s$, or equivalently $\sup_{Y \in A_s} \E{\trans{w_t^s(\Q,w)}\EQt{-Y}{s}} < \infty$.  By \cite[Theorem 1]{Y85} and $A_s$ decomposable, \[\sup_{Y \in A_s} \E{\trans{w_t^s(\Q,w)}\EQt{-Y}{s}} = \E{\esssup_{Y \in A_s} \trans{w_t^s(\Q,w)}\EQt{-Y}{s}}.\] 
Therefore $\esssup_{Y \in A_s} \trans{w_t^s(\Q,w)}\EQt{-Y}{s} \in \LoF{s}$ and, in particular, there exists some $u \in M_s$ so that $\trans{w_t^s(\Q,w)}u \geq \esssup_{Y \in A_s} \trans{w_t^s(\Q,w)}\EQt{-Y}{s}$ almost surely, i.e., $(\Q,w) \in \widehat \W_t^s$, see~\eqref{alpharep}.
Conversely, let $(\Q,w) \in \widehat \W_t^s$ and let $u \in M_s$ so that $u \in \alpha_s(\Q,w_t^s(\Q,w))$.  This yields $\trans{w_t^s(\Q,w)}u \geq \esssup_{Y \in A_s} \trans{w_t^s(\Q,w)}\EQt{-Y}{s}$ almost surely.  Thus, by \cite[Theorem 1]{Y85} and $A_s$ decomposable, 
\[\E{\trans{w_t^s(\Q,w)}u} \geq 
\sup_{Y \in A_s} \E{\trans{w_t^s(\Q,w)}\EQt{-Y}{s}}.\]  
Therefore $u \in \beta_s(\Q,w_t^s(\Q,w))$ and $(\Q,w) \in \W_t^s$.
\end{proof}

\begin{remark}
By the same logic as the proof of Corollary~\ref{cor_proper} and by Proposition~\ref{prop_cond_Wts}, $\recc{R_t(X)} = \bigcap_{(\Q,w) \in \W_t^t} \Gamma_t^M(w)$ for any $X \in \LdpF{}$ with $R_t(X) \neq \emptyset$. 
\end{remark}

\section{Proofs for Section~\ref{sec_prelim}}\label{appendix_prelim}

\subsection{Proof of Corollary~\ref{cor_cond_dual}}\label{proofCordual}
We will use the following proposition.
\begin{proposition}\label{prop_essinf_rep}
$\cl\lrsquare{A + \Gamma_t^M(w)} = \lrcurly{m \in M_t \; | \; \trans{w}m \geq \essinf_{a \in A} \trans{w}a}$ if $A \subseteq M_t$ is convex and decomposable.
\end{proposition}
\begin{proof}
First, note that 
\[D:=\lrcurly{m \in M_t \; | \; \trans{w}m \geq \essinf_{a \in A} \trans{w}a} = \overline\cl\lrsquare{A + \Gamma_t^M(w)},\]
where $\overline\cl$ is the closure with respect to the almost sure convergence of the set of random vectors.
Second, we will show that $\cl\lrsquare{A + \Gamma_t^M(w)} \subseteq D$ (recalling that the closure operator $\cl$ is w.r.t.\ the topological closure).  We will break this up into two cases: $p \in [1,\infty)$ and $p = \infty$.\\
Consider $p < \infty$.  To prove this statement we will show that $D$ is closed in the norm topology.  Let $(m_n)_{n \in \N} \to m \in M_t$ converge in the $p$-norm for $m_n \in D$ for every $n \in \N$.  Since $p$-norm convergence implies convergence in probability, we know that there exists a subsequence $(m_{n_k})_{k \in \N} \to m$ which converges almost surely, thus $m \in D$.  This implies that $\cl\lrsquare{A + \Gamma_t^M(w)} \subseteq D$.\\
Consider $p = \infty$.  To prove this statement we will show that $D$ is weak* closed.  
Let $(m_i)_{i \in I} \to m \in M_t$ in the weak* topology so that $m_i \in D$ for all $i \in I$.  Let $\Delta := \lrcurly{\omega \in \Omega \; | \; \trans{w(\omega)}m(\omega) < [\essinf_{a \in A} \trans{w}a](\omega)} \in \Ft{t}$.  Assume $\P(\Delta) > 0$, immediately it follows that
\begin{align*}
\E{1_\Delta \essinf_{a \in A}\trans{w}a} \leq \liminf_{i \in I} \E{1_\Delta \trans{w} m_i} = \E{1_\Delta \trans{w}m} < \E{1_\Delta \essinf_{a \in A} \trans{w}a}.
\end{align*}
By contradiction this implies $\P(\Delta) = 0$ and thus $\trans{w}m \geq \essinf_{a \in A} \trans{w}a$, i.e.\ $D$ is weak* closed.  This implies that $\cl\lrsquare{A + \Gamma_t^M(w)} \subseteq D$.  

Now we will show that $\cl\lrsquare{A + \Gamma_t^M(w)} \supseteq D$.  Assume this is not true, i.e.\ let $u \in D$ such that $u  \not\in \cl\lrsquare{A + \Gamma_t^M(w)}$.  By a separation argument there exists $v \in \LdqF{t}$ such that
\[\E{\trans{v}u} < \inf_{m \in \cl\lrsquare{A + \Gamma_t^M(w)}} \E{\trans{v}m}.\]
By construction of $\Gamma_t^M(w)$, 
\[\inf_{m \in \cl\lrsquare{A + \Gamma_t^M(w)}} \E{\trans{v}m} = \begin{cases}\inf_{a \in A} \E{\lambda \trans{w}a} &\text{if } v = \lambda w \text{ for some } \lambda \geq 0 \as\\ -\infty &\text{else}\end{cases}.\]
Noting that we can exchange the expectation and infimum due to decomposability (cf. \cite[Theorem 1]{Y85}),
\begin{align*}
\E{\lambda \trans{w}u} &< \inf_{a \in A} \E{\lambda \trans{w}a} = \E{\essinf_{a \in A} \lambda \trans{w}a} = \E{\lambda \essinf_{a \in A} \trans{w}a}.
\end{align*}
However this contradicts $u \in D$.  Thus $\cl\lrsquare{A + \Gamma_t^M(w)} \supseteq D$.
\end{proof}

\begin{proof}[Proof of Corollary~\ref{cor_cond_dual}]
First, by using Proposition~\ref{prop_essinf_rep}, we can reformulate the conditional penalty function as
\begin{equation}\label{alpharep}
\alpha_t(\Q,w) = \lrcurly{u \in M_t \; | \; \trans{w}u \geq \esssup_{Y \in A_t} \trans{w}\EQt{-Y}{t}}.
\end{equation}
So, for any $(\Q,w) \in \W_t$
\begin{align*}
&\lrparen{\EQt{-X}{t} + \Gamma_t(w)} \cap M_t -^. \alpha_t(\Q,w) 
\\&= \lrcurly{u \in M_t \; | \; u + \alpha_t(\Q,w) \subseteq \lrparen{\EQt{-X}{t} + \Gamma_t(w)} \cap M_t}\\
&= \lcurly{u \in M_t \; | \; \lrcurly{m \in M_t\;|\; \trans{w}m \geq \trans{w}u + \esssup_{Y \in A_t} \trans{w}\EQt{-Y}{t}} }
\\
&\quad\quad\rcurly{\subseteq \lrcurly{m \in M_t\;|\; \trans{w}m \geq \trans{w}\EQt{-X}{t}}}\\
&= \lrcurly{u \in M_t \; | \; \trans{w}u + \esssup_{Y \in A_t} \trans{w}\EQt{-Y}{t} \geq \trans{w}\EQt{-X}{t}}\\
&= \lrcurly{u \in M_t \; | \; \trans{w}u \geq \essinf_{Y \in A_t} \trans{w}\EQt{Y - X}{t}}\\
&= -\bar\alpha_t(\Q,w) + \lrparen{\EQt{-X}{t} + \Gamma_t(w)} \cap M_t.
\end{align*}
Thus, now the result follows directly from the dual representation~\cite[Theorem 2.3]{FR12b} w.r.t.\ the negative conjugate function.  The above chain of equalities follow via the definition of the Minkowski subtraction, the result of Proposition~\ref{prop_essinf_rep}, reformulating the inclusion, and by the definition of the negative conjugate function from~\cite{FR12b} respectively.  The conditionally coherent case is as in~\cite[Corollary~2.4]{FR12b}.
\end{proof}

\subsection{Proof of Theorem~\ref{thm_penalty}}
The proof of Theorem~\ref{thm_penalty} involves the following two propositions.
\begin{proposition}\label{prop_penalty1}
Consider $w\in \plus{M_{t,+}} \backslash \prp{M_t}$. For any sets $A,B \in \mathcal{G}(M_t;M_{t,+})$ such that $A \neq \emptyset$ if $\cl\lrsquare{B + G_t^M(w)} = M_t$, and  $B\neq \emptyset$ if $\cl\lrsquare{A + G_t^M(w)} = M_t$, it holds that
\begin{equation}\label{eqproofA1}
G_t^M(w) -^. \cl\lrsquare{A + B} = \cl\lrsquare{\lrparen{G_t^M(w) -^. A} + \lrparen{G_t^M(w) -^. B}}.
\end{equation}
\end{proposition}
\begin{proof}
``$\supseteq$'' If  $\cl\lrsquare{\lrparen{G_t^M(w) -^. A} + \lrparen{G_t^M(w) -^. B}}=\emptyset$, the inclusion is trivial. So, let us assume the right hand side of \eqref{eqproofA1} is nonempty and consider an element $u \in \cl\lrsquare{\lrparen{G_t^M(w) -^. A} + \lrparen{G_t^M(w) -^. B}}$.  Then there exists a net $(u_i^A)_{i \in I}$ and $(u_j^B)_{j \in J}$ such that $u = \lim_{i,j} \lrparen{u_i^A + u_j^B}$ and $u_i^A \in G_t^M(w) -^. A$ and $u_j^B \in G_t^M(w) -^. B$ for every $i,j$.  Immediately, by definition of $-^.$ we obtain
\[u_i^A + u_j^B + \cl\lrsquare{A + B} \subseteq G_t^M(w) + G_t^M(w) = G_t^M(w).\]
Therefore, $u_i^A + u_j^B \in G_t^M(w) -^. \cl\lrsquare{A + B}$ for every $i,j$.  Since $G_t^M(w) -^. \cl\lrsquare{A + B}$ is closed by definition, one has $u \in G_t^M(w) -^. \cl\lrsquare{A + B}$.

``$\subseteq$'' If $G_t^M(w) -^. \cl\lrsquare{A + B}=\emptyset$, the inclusion is trivial. So, let us assume the left hand side of \eqref{eqproofA1} is nonempty and consider $u \in G_t^M(w) -^. \cl\lrsquare{A + B}$.  By definition this is equivalent to 
$\E{\trans{w}u} + \inf_{a \in A} \E{\trans{w}a} + \inf_{b \in B} \E{\trans{w}b} \geq 0$ and $u \in M_t$.
First consider the case where $\inf_{b \in B} \E{\trans{w}b} \in \R$. Let $u^B \in M_t$ so that $\E{\trans{w}u^B} = -\inf_{b \in B}\E{\trans{w}b}$, where the existence of $u^B$ is guaranteed by the continuity of the linear operator, $M_t$ being a linear space, and $w\not\in \prp{M_t}$. Define $u^A := u - u^B \in M_t$.  By construction $u^A \in G_t^M(w) -^. A$, $u^B \in G_t^M(w) -^. B$, and $u = u^A + u^B$.  That is, $u \in \cl\lrsquare{\lrparen{G_t^M(w) -^. A} + \lrparen{G_t^M(w) -^. B}}$. 
If $\inf_{b \in B} \E{\trans{w}b} = \infty$ then by assumption $\inf_{a \in A} \E{\trans{w}a} > -\infty$, so $G_t^M(w) -^. \cl\lrsquare{A + B} = M_t$, $G_t^M(w) -^. A \neq \emptyset$, and $G_t^M(w) -^. B = M_t$ and thus the inclusion trivially holds. The case $\inf_{b \in B} \E{\trans{w}b} = -\infty$ cannot occur under the current assumption of the left hand side of \eqref{eqproofA1} being nonempty as it would imply $\inf_{a \in A} \E{\trans{w}a} < \infty$ by assumption and therefore $G_t^M(w) -^. \cl\lrsquare{A+B} = \emptyset$.
\end{proof}

\begin{proposition}\label{prop_penalty2}
Let $s > t$ and $(\Q,w) \in \W_t$. For any set $A \in \mathcal{G}(M_s;M_{s,+})$ it holds \[G_t^M(w) -^. \EQt{A}{t} = \EQt{G_s^M(w_t^s(\Q,w)) -^. A}{t}.\]
\end{proposition}
\begin{proof}
\begin{align*}
&\EQt{G_s^M(w_t^s(\Q,w)) -^. A}{t} = \lrcurly{\EQt{u}{t}\;|\; u \in M_s, u + A \subseteq G_s^M(w_t^s(\Q,w))}\\
    &= \lrcurly{\EQt{u}{t}\;|\; u \in M_s, \inf_{a \in A} \E{\trans{w_t^s(\Q,w)}(u + a)} \geq 0}\\
    &= \lrcurly{\EQt{u}{t}\;|\; u \in M_s, \inf_{a \in A} \E{\trans{w}\EQt{u + a}{t}} \geq 0}\\
    &= \lrcurly{u \in M_t\;|\; \inf_{a \in A} \E{\trans{w}(\EQt{a}{t} + u)} \geq 0}\\
    &= \lrcurly{u \in M_t\;|\; u + \EQt{A}{t} \subseteq G_t^M(w)}= G_t^M(w) -^. \EQt{A}{t}.
\end{align*}
\end{proof}

\begin{proof}[Proof of Theorem~\ref{thm_penalty}]
Recall from \cite[Remark~5.5]{setOPsurvey} that the conjugate and negative conjugate are related via $\beta_t(\Q,w) = G_t^M(w) -^. (-\bar\beta_t(\Q,w))$ and $-\bar\beta_t(\Q,w) = G_t^M(w) -^. \beta_t(\Q,w)$ for every $(\Q,w) \in \W_t$.   Recall from \cite[Theorem 3.2]{FR12b} that multiportfolio time consistency is equivalent to the cocycle condition on the negative conjugates, i.e.,
\[-\bar\beta_t(\Q,w) = \cl\lrparen{-\bar\beta_{t,s}(\Q,w) + \EQt{-\bar\beta_s(\Q,w_t^s(\Q,w))}{t}}\]
for every $(\Q,w) \in \W_t$ and all times $0 \leq t < s \leq T$.  Let $(\Q,w) \in \W_t$.

``$\Rightarrow$'' 
Since $\seq{R}$ is multiportfolio time consistent, we obtain
\begin{align*}
\beta_t&(\Q,w) = G_t^M(w) -^. (-\bar\beta_t(\Q,w))\\
    &= G_t^M(w) -^. \cl\lrparen{-\bar\beta_{t,s}(\Q,w) + \EQt{-\bar\beta_s(\Q,w_t^s(\Q,w))}{t}}\\
    &= \cl\lrsquare{\lrparen{G_t^M(w) -^. (-\bar\beta_{t,s}(\Q,w))} + \lrparen{G_t^M(w) -^. \EQt{-\bar\beta_s(\Q,w_t^s(\Q,w))}{t}}}\\
    &= \cl\lrsquare{\beta_{t,s}(\Q,w) + \lrparen{\EQt{G_s^M(w_t^s(\Q,w) -^. (-\bar\beta_s(\Q,w_t^s(\Q,w)))}{t}}}\\
    &= \cl\lrsquare{\beta_{t,s}(\Q,w) + \EQt{\beta_s(\Q,w_t^s(\Q,w))}{t}},
\end{align*}
where the third and fourth equations follow from Proposition~\ref{prop_penalty1} and Proposition~\ref{prop_penalty2}, respectively. The assumptions of Proposition~\ref{prop_penalty1} are satisfied as $-\bar\beta_{t}(\Q,w)) \neq \emptyset$ for $(\Q,w) \in \W_t$ and thus $A:=-\bar\beta_{t,s}(\Q,w)\neq\emptyset$ and $B:= \EQt{-\bar\beta_s(\Q,w_t^s(\Q,w))}{t}\neq\emptyset$. 

``$\Leftarrow$'' 
Since $\beta_t(\Q,w) = \cl\lrparen{\beta_{t,s}(\Q,w) + \EQt{\beta_s(\Q,w_t^s(\Q,w))}{t}}$, it holds
\begin{align*}
-\bar\beta_t&(\Q,w) = G_t^M(w) -^. \beta_t(\Q,w)\\
    &= G_t^M(w) -^. \cl\lrparen{\beta_{t,s}(\Q,w) + \EQt{\beta_s(\Q,w_t^s(\Q,w))}{t}}\\
    &= \cl\lrsquare{\lrparen{G_t^M(w) -^. \beta_{t,s}(\Q,w)} + \lrparen{G_t^M(w) -^. \EQt{\beta_s(\Q,w_t^s(\Q,w))}{t}}}\\
    &= \cl\lrsquare{-\bar\beta_{t,s}(\Q,w) + \lrparen{\EQt{G_s^M(w_t^s(\Q,w)) -^. \beta_s(\Q,w_t^s(\Q,w))}{t}}}\\
    &= \cl\lrparen{-\bar\beta_{t,s}(\Q,w) + \EQt{-\bar\beta_s(\Q,w_t^s(\Q,w))}{t}}.
\end{align*}
Again, the third and fourth equations follow from Proposition~\ref{prop_penalty1} and Proposition~\ref{prop_penalty2}, respectively. The assumptions of Proposition~\ref{prop_penalty1} are satisfied as $\beta_t(\Q,w) \neq M_t$ for $(\Q,w) \in \W_t$ and thus by the cocycle condition $A:=  \beta_{t,s}(\Q,w)=\cl\lrsquare{A + G_t^M(w)} \neq M_t $ and $B:=\EQt{\beta_s(\Q,w_t^s(\Q,w))}{t}=\cl\lrsquare{B + G_t^M(w)}  \neq M_t$. 
\end{proof}

\subsection{Proof of Theorem~\ref{thm_cond_penalty}}
The proof of Theorem~\ref{thm_cond_penalty} uses the following two propositions.

\begin{proposition}\label{prop_cond_penalty1}
Let $w\in \plus{M_{t,+}} \backslash \prp{M_t}$. For any decomposable sets $A,B \in \mathcal{G}(M_t;M_{t,+})$
with $A \neq \emptyset$ if $\P(\cl\lrsquare{B + \Gamma_t^M(w)} = M) > 0$, and $B \neq \emptyset$ if $\P(\cl\lrsquare{A + \Gamma_t^M(w)} = M) > 0$, it holds that
\begin{equation}\label{plus}
\Gamma_t^M(w) -^. \cl\lrsquare{A + B} = \cl\lrsquare{\lrparen{\Gamma_t^M(w) -^. A} + \lrparen{\Gamma_t^M(w) -^. B}}.
\end{equation}
\end{proposition}
\begin{proof}
This follows analogously to the proof of Proposition~\ref{prop_penalty1} noting that, for the second part of the proof, $\essinf_{b \in B} \trans{w}b \not\in \LoF{t}$ if and only if either $B = \emptyset$ or $\P(\cl\lrsquare{B + \Gamma_t^M(w)} = M) > 0$.
\end{proof}

\begin{proposition}\label{prop_cond_penalty2}
Let $s > t$ and $(\Q,w) \in \W_t^e$. For any decomposable set $A \in \mathcal{G}(M_s;M_{s,+})$ it holds that
\[\Gamma_t^M(w) -^. \EQt{A}{t} = \cl\EQt{\Gamma_s^M(w_t^s(\Q,w)) -^. A}{t}.\]
\end{proposition}
\begin{proof}
\begin{align*}
&\cl\EQt{\Gamma_s^M(w_t^s(\Q,w)) -^. A}{t} = \cl\lrcurly{\EQt{u}{t}\;|\; u \in M_s, u + A \subseteq \Gamma_s^M(w_t^s(\Q,w))}\\
    &= \cl\lrcurly{\EQt{u}{t}\;|\; u \in M_s, \essinf_{a \in A} \trans{w_t^s(\Q,w)}(u + a) \geq 0}\\
    &\subseteq \cl\lrcurly{\EQt{u}{t}\;|\; u \in M_s, \Et{\essinf_{a \in A} \trans{w_t^s(\Q,w)}(u + a)}{t} \geq 0}\\
    &= \cl\lrcurly{\EQt{u}{t}\;|\; u \in M_s, \essinf_{a \in A} \trans{w}\EQt{u + a}{t} \geq 0}\\
    &= \cl\lrcurly{u \in M_t\;|\; u + \EQt{A}{t} \subseteq \Gamma_t^M(w)}
    = \Gamma_t^M(w) -^. \EQt{A}{t}.
\end{align*}
We can interchange expectation and essential infimum because $A$ is decomposable and a set of integrable random variables (see \cite[Theorem 1]{Y85}).   If $\Gamma_t^M(w) -^. \EQt{A}{t}$ is empty, then equality follows immediately. Now consider a point $u \in \Gamma_t^M(w) -^. \EQt{A}{t}$ and assume $u \not\in \cl\EQt{\Gamma_s^M(w_t^s(\Q,w)) -^. A}{t}$.  Since $\cl\EQt{\Gamma_s^M(w_t^s(\Q,w)) -^. A}{t}$ is closed and convex, we can separate $\{u\}$ and $\cl\EQt{\Gamma_s^M(w_t^s(\Q,w)) -^. A}{t}$ by some $v \in \LdqF{t}$, i.e. 
\begin{align*}
\E{\trans{v}u} &< \inf_{z \in \cl\EQt{\Gamma_s^M(w_t^s(\Q,w)) -^. A}{t}} \E{\trans{v}z} = \inf_{z \in \Gamma_s^M(w_t^s(\Q,w)) -^. A} \E{\trans{w_t^s(\Q,v)}z}\\
&=\E{\essinf_{z \in \Gamma_s^M(w_t^s(\Q,w)) -^. A} \trans{w_t^s(\Q,v)}z}.
\end{align*}
Note that in the last equality above we can interchange the expectation and infimum since $\Gamma_s^M(w_t^s(\Q,w)) -^. A$ is decomposable and integrable. 
By construction
\[\essinf_{z \in \Gamma_s^M(w_t^s(\Q,w)) -^. A} \trans{w_t^s(\Q,v)}z = \begin{cases} \esssup_{a \in A} (-\trans{w_t^s(\Q,w)}a )&\text{on } D\\ -\infty &\text{on } D^c,\end{cases}\]
where $D = \lrcurly{\omega \in \Omega \; | \; G_0(w_t^s(\Q,v)[\omega]) = G_0(w_t^s(\Q,w)[\omega])}$.  Since $\Q \in \mathcal{M}^e$, one has that $G_0(w_t^s(\Q,v)[\omega]) = G_0(w_t^s(\Q,w)[\omega])$ if and only if $v(\omega) = \lambda(\omega)w(\omega)$ for some $\lambda \in \LpK{0}{t}{\R_{++}}$ (such that $\lambda w \in \LqF{t}$).  Thus, 
$\E{\essinf_{z \in \Gamma_s^M(w_t^s(\Q,w)) -^. A} \trans{w_t^s(\Q,v)}z} > -\infty$ if and only if
\begin{align*}
\E{\essinf_{z \in \Gamma_s^M(w_t^s(\Q,w)) -^. A} \trans{w_t^s(\Q,v)}z} &= \E{\lambda \essinf_{z \in \Gamma_s^M(w_t^s(\Q,w)) -^. A} \trans{w_t^s(\Q,w)}z} 
\\
&= \E{\lambda \esssup_{a \in A}( -\trans{w}\EQt{a}{t})}.
\end{align*}
But this implies $\E{\lambda \trans{w}u} < \E{\lambda \esssup_{a \in A} (-\trans{w}\EQt{a}{t})}$, which is a contradiction to $u \in \Gamma_t^M(w) -^. \EQt{A}{t}$.
\end{proof}

\begin{proof}[Proof of Theorem~\ref{thm_cond_penalty}]
Note that $\alpha_t(\Q,w) = \Gamma_t^M(w) -^. (-\bar\alpha_t(\Q,w))$ and $-\bar\alpha_t(\Q,w) = \Gamma_t^M(w) -^. \alpha_t(\Q,w)$ for every $(\Q,w) \in \W_t$.  
The proof is analog to the proof of Theorem~\ref{thm_penalty}, but using Propositions~\ref{prop_cond_penalty1} and~\ref{prop_cond_penalty2} instead, where the assumptions of Proposition~\ref{prop_cond_penalty1} are satisfied as for 
$(\Q,w) \in \W_t^e$, $-\bar\alpha_t(\Q,w)\neq\emptyset$ 
and $\P(\alpha_t(\Q,w)= M)=0$. 
\end{proof}

\section{Proofs for Section~\ref{sec_supermtg}}

\subsection{Proof of Lemma~\ref{lemma1}}
Throughout we will use the following set of dual variables for times $t \leq s \leq T$ \[\W_t^s := \lrcurly{(\Q,w) \in \W_t\;|\; \beta_s(\Q,w_t^s(\Q,w)) \neq \emptyset}.\]
\begin{proof}
``$\Leftarrow$'' We will prove first that conditions~\eqref{eq_supermtg-1} and~\eqref{eq_supermtg-2} imply the supermartingale property \eqref{SMP}.
    If $(\Q,w) \not\in \W_t^t$, then $V_t^{(\Q,w)}(X) = \emptyset$ for any $X \in \LdpF{}$ and \eqref{SMP} is satisfied.  Now, let $X \in \LdpF{}$ and $(\Q,w) \in \W_t^t$. It holds
    \begin{align*}
    V_t^{(\Q,w)}(X) 
    = \lrcurly{u \in M_t\;|\; \E{\trans{w}u} \geq \inf_{Z \in R_t(X)} \E{\trans{w}Z} + \sup_{Y_t \in A_t} \E{\trans{w}\EQt{-Y_t}{t}}}.
    \end{align*}
    Similarly, it follows that
        \begin{align*}
            &\EQt{V_s^{(\Q,w_t^s(\Q,w))}(X)}{t} 
            \\
            =& \lrcurly{u \in M_t\;|\; \E{\trans{w}u} \geq \inf_{Z \in R_s(X)} \E{\trans{w}\EQt{Z}{t}} + \sup_{Y_s \in A_s} \E{\trans{w}\EQt{-Y_s}{t}}}
        \end{align*}
        if $\beta_s(\Q,w_t^s(\Q,w)) \neq \emptyset$, i.e., if $(\Q,w) \in \W_t^s$.  The latter is true as we will now show that condition~\eqref{eq_supermtg-1} yields in fact $\W_t^t \subseteq \W_t^s$.
       Note that condition~\eqref{eq_supermtg-1} implies $A_t \supseteq A_s + A_{t,s}$ (Lemma 3.6(iii) in~\cite{FR12}) which yields for every $(\Q,w)\in\W_t$
        \begin{align*}
            \beta_t(\Q,w) &\subseteq \cl\lrparen{\beta_{t,s}(\Q,w) +
            \EQt{\beta_s(\Q,w_t^s(\Q,w))}{t}}.
        \end{align*}
        This follows from a trivial modification to the first part of the proof of Theorem~3.2 in~\cite{FR12b}, followed by switching to the positive conjugate via $\beta_t(\Q,w) = G_t^M(w) -^. (-\bar\beta_t(\Q,w))$ and using Propositions~\ref{prop_penalty1} and~\ref{prop_penalty2}.
        Therefore, if $\beta_t(\Q,w) \neq \emptyset$ it must follow that $\beta_s(\Q,w_t^s(\Q,w)) \neq
        \emptyset$, i.e.\ $\W_t^t \subseteq \W_t^s$.
        Thus, the supermartingale property holds via
        \begin{align}
            \nonumber & V_t^{(\Q,w)}(X) = \lrcurly{u \in M_t\;|\; \E{\trans{w}u} \geq \inf_{Z_t \in R_t(X)} \sup_{Y_t \in A_t} \E{\trans{w}(Z_t - \EQt{Y_t}{t})}}\\
            \label{eq_supermtg_penaltysum} &\subseteq \lrcurly{u \in M_t\;|\; \E{\trans{w}u} \geq \inf_{Z_t \in R_t(X)} \sup_{\substack{Y_{t,s} \in A_{t,s}\\ Y_s \in A_s}} \E{\trans{w}(Z_t - \EQt{Y_{t,s} + Y_s}{t})}}\\
            \begin{split} \nonumber 
                &= \lrcurly{u \in M_t\;|\; \E{\trans{w}u} \geq \inf_{Z_t \in R_t(X)} \sup_{Y_{t,s} \in A_{t,s}} \E{\trans{w}(Z_t - \EQt{Y_{t,s}}{t})}}\\
                &\quad\quad + \lrcurly{u \in M_t\;|\; \E{\trans{w}u} \geq \sup_{Y_s \in A_s} \E{-\trans{w}\EQt{Y_s}{t}}}
                \end{split}\\
            \begin{split} \label{eq_supermtg_recursive} 
                &\subseteq \lrcurly{u \in M_t\;|\; \E{\trans{w}u} \geq \inf_{Z_s \in R_s(X)} \E{\trans{w}\EQt{Z_s}{t}}}\\
                &\quad\quad + \lrcurly{u \in M_t\;|\; \E{\trans{w}u} \geq \sup_{Y_s \in A_s} \E{-\trans{w}\EQt{Y_s}{t}}}
                \end{split}\\
            \nonumber &= \lrcurly{u \in M_t\;|\; \E{\trans{w}u} \geq \inf_{Z_s \in R_s(X)} \sup_{Y_s \in A_s} \E{\trans{w}\EQt{Z_s-Y_s}{t}}}\\
            \nonumber &= \EQt{V_s^{(\Q,w_t^s(\Q,w))}(X)}{t}.
        \end{align}
        Inclusion~\eqref{eq_supermtg_penaltysum} follows from condition~\eqref{eq_supermtg-1} (as \eqref{eq_supermtg-1} implies $A_t \supseteq A_s + A_{t,s}$ by Lemma 3.6(iii) in~\cite{FR12}).
        Inclusion~\eqref{eq_supermtg_recursive} is true if and only if
        \[\inf_{Z_t \in R_t(X)} \E{\trans{w}Z_t} \geq \inf_{Z_s \in R_s(X)} \inf_{Y_{t,s} \in A_{t,s}} \E{\trans{w}\EQt{Z_s + Y_{t,s}}{t}}.\]
        But this follows from $R_t(X) \subseteq \cl \bigcup_{Z \in R_s(X)} ((\EQt{Z}{t} + G_t(w)) \cap M_t -^. \beta_{t,s}(\Q,w))$, which is immediate by condition~\eqref{eq_supermtg-2}.

\medskip

``$\Rightarrow$'' We will now prove that the supermartingale property \eqref{SMP} implies~\eqref{eq_supermtg-1} and~\eqref{eq_supermtg-2}.  First note that \eqref{SMP} yields $\W_t^t \subseteq \W_t^s$:  Assume $\beta_s(\Q,w_t^s(\Q,w)) = \emptyset$, then $V_s^{(\Q,w_t^s(\Q,w))}(0) = \emptyset$ by definition of the Minkowski addition.  This implies $V_t^{(\Q,w)}(0) = \emptyset$ by the supermartingale relation.  However this can only occur if $\beta_t(\Q,w) = \emptyset$ since $R_t(0) \neq \emptyset$ by definition.
      
      We will now show that the supermartingale property \eqref{SMP} 
         implies time consistency, which then yields condition~\eqref{eq_supermtg-1}.
        Assume $R_s(X) \subseteq R_s(Y)$ for some $X,Y \in \LdpF{}$. We need to show that $R_t(X) \subseteq R_t(Y)$.
            Let $(\Q,w) \in \W_t^t$. By $\W_t^t \subseteq \W_t^s$, 
            it follows that $\beta_s(\Q,w_t^s(\Q,w)) \neq \emptyset$.  Also assume $R_t(X) \neq \emptyset$ (if it were empty then it would trivially follow that $R_t(X) \subseteq R_t(Y)$). It holds
            \begin{align}
               \nonumber &\cl\lrsquare{R_t(X) + \beta_t(\Q,w)} = V_t^{(\Q,w)}(X)
               \subseteq \EQt{V_s^{(\Q,w_t^s(\Q,w))}(X)}{t}\\
                \nonumber  &= \EQt{\cl\lrsquare{R_s(X) + \beta_s(\Q,w_t^s(\Q,w))}}{t}
               \subseteq \EQt{\cl\lrsquare{R_s(Y) + \beta_s(\Q,w_t^s(\Q,w))}}{t}\\
                \nonumber  &\subseteq \mathbb{E}^{\Q}\lsquare{\cl\lsquare{\lrparen{(\EQt{-Y}{s} + G_s(w_t^s(\Q,w))) \cap M_s -^. \beta_s(\Q,w_t^s(\Q,w))}}}\\
                \nonumber  &\quad\quad\quad \rsquare{\left. \rsquare{+ \beta_s(\Q,w_t^s(\Q,w))} \right| \Ft{t}}\\
                \nonumber&
                \subseteq \EQt{(\EQt{-Y}{s} + G_s(w_t^s(\Q,w))) \cap M_s}{t}
                = (\EQt{-Y}{t} + G_t(w)) \cap M_t.
            \end{align}
            The last equality follows from the tower property and \cite[Corollary~A.4]{FR12b}.
            The last inclusion follows from 
            $\cl[([A+G_{s}^M] -^. B) + B] \subseteq \cl[A + G_{s}^M]$,
            which holds by the definition of the Minkowski subtraction.
         Here, we have $A = (\EQt{-Y}{s} + G_s(w_t^s(\Q,w))) \cap M_s$, $B = \beta_s(\Q,w_t^s(\Q,w))$, and notice that $\cl[A + G_s^M] = A$,  where we used the notation $G_s^M= G_s^M(w_t^s(\Q,w))$. 

From the above we have
            \begin{equation}
            \label{proofl1}
          R_t(X) + \beta_t(\Q,w)\subseteq\cl[R_t(X) + \beta_t(\Q,w)] \subseteq (\EQt{-Y}{t} + G_t(w)) \cap M_t,
          \end{equation}
          which yields 
            \begin{align}
               \label{eqprooftc0} R_t(X) &\subseteq R_t(X) + G_t^M(w)
               \subseteq (R_t(X) + \beta_t(\Q,w)) -^. \beta_t(\Q,w)\\
               \label{eqprooftc} &\subseteq (\EQt{-Y}{t} + G_t(w)) \cap M_t -^. \beta_t(\Q,w)
            \end{align}
            for any $(\Q,w) \in \W_t$ (as it trivially holds for those $(\Q,w) \in \W_t$ for which $\beta_t(\Q,w) = \emptyset$ as well).
            The second inclusion in \eqref{eqprooftc0} follows from 
          $A + G_t^M(w)\subseteq (A + B + G_t^M(w)) -^. B$,
            which holds by the definition of the Minkowski subtraction.
            Here, $A = R_t(X)$, $B = \beta_t(\Q,w)$, and $B + G_t^M(w) = B$. 

           The inclusions \eqref{eqprooftc0},~\eqref{eqprooftc} yield \[R_t(X) \subseteq \bigcap_{(\Q,w) \in \W_t} \lrsquare{(\EQt{-Y}{t} + G_t(w)) \cap M_t -^. \beta_t(\Q,w)} = R_t(Y).\]
This is time consistency which implies~\eqref{eq_supermtg-1} as follows. Let $Y\in \bigcup_{Z \in R_s(X)} R_t(-Z)$, i.e.\ there exists a $\hat Z\in R_s(X)$ such that $Y\in R_t(-\hat Z)$. We need to show that $Y\in R_t(X)$. By translativity and normalization of the risk measure, it holds that \[R_s(-\hat Z)=R_s(0)+\hat Z\subseteq R_s(0)+R_s(X)=R_s(X).\] Time consistency now yields $R_t(-\hat Z)\subseteq R_t(X)$ and as $Y\in R_t(-\hat Z)$, it holds that $Y\in R_t(X)$ and thus~\eqref{eq_supermtg-1}.

        We will now prove that the supermartingale property \eqref{SMP} implies~\eqref{eq_supermtg-2}.  Let  $(\Q,w) \in \W_t^t$. By \eqref{eqprooftc0},
        
        \begin{align}
        \nonumber &R_t(X) \subseteq 
        (R_t(X) + \beta_t(\Q,w)) -^. \beta_t(\Q,w) \subseteq \cl[R_t(X) + \beta_t(\Q,w)] -^. \beta_t(\Q,w)\\
        \nonumber &= V_t^{(\Q,w)}(X) -^. \beta_t(\Q,w)
        \subseteq \EQt{V_s^{(\Q,w_t^s(\Q,w))}(X)}{t} -^. \beta_t(\Q,w)\\
        \nonumber &=\EQt{\cl[R_s(X) + \beta_s(\Q,w_t^s(\Q,w))]}{t} -^. \beta_t(\Q,w)\\
        \nonumber &= \EQt{\cl\lrparen{\bigcup_{Z \in R_s(X)} \lrsquare{R_s(-Z) + \beta_s(\Q,w_t^s(\Q,w))}}}{t} -^. \beta_t(\Q,w)\\
        \nonumber &\subseteq \cl\lrparen{\bigcup_{Z \in R_s(X)} \lrsquare{\EQt{R_s(-Z) + \beta_s(\Q,w_t^s(\Q,w))}{t} -^. \beta_t(\Q,w)}}\\
        \label{incll1} &\subseteq \cl\lrparen{\bigcup_{Z \in R_s(X)} \lrsquare{\EQt{(\EQt{Z}{s} + G_s(w_t^s(\Q,w)))\cap M_s}{t} -^. \beta_t(\Q,w)}}\\
        \nonumber &= \cl\lrparen{\bigcup_{Z \in R_s(X)} \lrsquare{\EQt{Z}{t} + G_t(w))\cap M_t -^. \beta_t(\Q,w)}}.
        \end{align}

        Inclusion~\eqref{incll1} holds as for any $\tilde s>s$, $R_{\tilde s}(-Z)\subseteq R_{\tilde s}(-Z)$, which yields by~\eqref{proofl1} (setting $X=Y=-Z$) that 
       \[R_s(-Z) + \beta_s(\Q,w_t^s(\Q,w)) \subseteq ( \EQt{Z}{s} + G_s(w_t^s(\Q,w))) \cap M_s.\]
       The last equality follows from the tower property and \cite[Corollary A.4]{FR12b}.
        
        From the chain of inclusions above, we therefore conclude that for all dual variables in $(\Q,w)\in \W_t$ (as it trivially holds also for those $(\Q,w)\in \W_t$ that are not in $\W_t^t$)
        \begin{equation}
        \label{equ_l1}
        R_t(X) \subseteq \bigcap_{(\Q,w) \in \W_t} \cl \bigcup_{Z \in R_s(X)} \lrsquare{\lrparen{\EQt{Z}{t} + G_t(w)} \cap M_t -^. \beta_t(\Q,w)}.
        \end{equation}

        To prove~\eqref{eq_supermtg-2}, it remains to show that we can replace $\beta_t(\Q,w)$ in \eqref{equ_l1} by the stepped version $\beta_{t,s}(\Q,w)$. 
        Trivially, the inequality
        \begin{equation}\label{ineq_supermtg-proof}
        \inf_{Y_{t,s} \in A_{t,s}} \E{\transp{w+m_{\perp}}\EQt{Y_{t,s} - Z}{t}} \geq \inf_{Y_t \in A_t} \E{\transp{w+m_{\perp}}\EQt{Y_t-Z}{t}}
        \end{equation}
        holds for all $(\Q,m_{\perp}) \in \W_t(w)$.  The reverse is not true in general; this is the difficult part of the proof.
        Let $Z \in M_s$.  And let $\rho_{t,s}(Z) := \inf_{u \in R_{t,s}(Z)} \E{\trans{w}u}$ for $w \in \plus{\recc{R_t(0)}}\backslash \prp{M_t}$, which is a proper, convex, lower semicontinuous function (see Proposition~\ref{prop_lsc} and Corollary~\ref{rem_proper}) with representation given in Corollary~\ref{cor_scalar_stepped}. Note that the first and last lines below follow from $\rho_t(Z) = -\infty$ for $w \not\in \plus{\recc{R_t(0)}}$ by Corollary~\ref{rem_proper}. 
        Also, note that $R_{t,s}(0) = R_t(0)$ by definition. The representation in the first and last lines follow from a standard separation argument.
        \begin{align*}
        R_{t,s}(Z) &= \bigcap_{w \in \plus{\recc{R_t(0)}} \backslash \prp{M_t}} \lrcurly{m \in M_t \; | \; \E{\trans{w}m} \geq \rho_{t,s}(Z)}\\
        &= \bigcap_{w \in \plus{\recc{R_t(0)}} \backslash \prp{M_t}} \lcurly{m \in M_t \; | \; \E{\trans{w}m} \geq }
        \\&\quad\quad\quad\quad\quad\quad\quad\quad
       \rcurly{ \sup_{(\Q,m_{\perp}) \in \W_{t,s}(w)} \inf_{Y_{t,s} \in A_{t,s}} \E{\transp{w + m_{\perp}}\EQt{Y_{t,s} - Z}{t}}}\\ 
        &\subseteq \bigcap_{w \in \plus{\recc{R_t(0)}} \backslash \prp{M_t}} \lcurly{m \in M_t \; | \; \E{\trans{w}m} \geq }
        \\&\quad\quad\quad\quad\quad\quad\quad\quad
        \rcurly{\sup_{(\Q,m_{\perp}) \in \W_t(w)} \inf_{Y_t \in A_t} \E{\transp{w + m_{\perp}}\EQt{Y_t - Z}{t}}}\\
        &= \bigcap_{w \in \plus{\recc{R_t(0)}} \backslash \prp{M_t}} \lrcurly{m \in M_t \; | \; \E{\trans{w}m} \geq \rho_t(Z)}\\
        &= R_t(Z) = R_{t,s}(Z).
        \end{align*}
        Thus, the above inclusion is actually an equality.
        This implies a weaker form of the reverse of inequality~\eqref{ineq_supermtg-proof} that will be enough to obtain the desired replacement of $\beta_t(\Q,w)$ in \eqref{equ_l1} by the stepped version $\beta_{t,s}(\Q,w)$:
        One obtains that for every $Z \in M_s$ and every $w \in \plus{\recc{R_t(0)}}\backslash \prp{M_t}$ it holds that for all $(\Q,m_{\perp}) \in \W_t(w)$ there exists an $(\R,n_{\perp}) \in \W_t(w)$ such that
        \[\inf_{Y_{t,s} \in A_{t,s}} \E{\transp{w+m_{\perp}}\EQt{Y_{t,s} - Z}{t}} \leq \inf_{Y_t \in A_t} \E{\transp{w+n_{\perp}}\ERt{Y_t-Z}{t}}.\]
        This is because every such constraint is ``active'', i.e., if any were made any stricter it would shrink the set $R_{t,s}(Z)$.
        In particular, this is true if $m_{\perp} = 0 \in \prp{M_t}$. 
        Additionally, for $w \not\in \plus{\recc{R_t(0)}} \backslash \prp{M_t}$, $\inf_{Y_{t,s} \in A_{t,s}} \E{\transp{w+m_{\perp}}\EQt{Y_{t,s}-Z}{t}} = -\infty$ 
        for any $(\Q,m_{\perp}) \in \W_t(w)$.
        Thus, we can conclude for every $Z \in M_s$ it holds that for all $ (\Q,w) \in \W_t$ there exists an $(\R,v) \in \W_t$ such that $v \in w + \prp{M_t}$ and
        \[\inf_{Y_{t,s} \in A_{t,s}} \E{\trans{w}\EQt{Y_{t,s} - Z}{t}} \leq \inf_{Y_t \in A_t} \E{\trans{v}\ERt{Y_t - Z}{t}}.\]
        In particular, for every $(\Q,w) \in \W_t$ there exists some $\R(\Q,w,Z) \in \mathcal{M}$ and $v(\Q,w,Z) \in w + \prp{M_t}$ so that $(\R(\Q,w,Z),v(\Q,w,Z)) \in \W_t$ and
         \begin{align*}
          &\lrparen{\EQt{-Z}{t} + G_t(w)}\cap M_t -^. \beta_{t,s}(\Q,w) \supseteq\\
          &\quad\quad\quad  \lrparen{\EPt{\R(\Q,w,Z)}{-Z}{t} + G_t(v(\Q,w,Z))} \cap M_t -^. \beta_t(\R(\Q,w,Z),v(\Q,w,Z)).
          \end{align*}
        Using \eqref{equ_l1}, this implies~\eqref{eq_supermtg-2}:
        \begin{align*}
        R_t(X) &\subseteq \bigcap_{(\Q,w) \in \W_t} \cl \bigcup_{Z \in R_s(X)} \lrsquare{\lrparen{\EQt{Z}{t} + G_t(w)} \cap M_t -^. \beta_t(\Q,w)}\\
        &\subseteq \bigcap_{(\Q,w) \in \W_t} \cl \bigcup_{Z \in R_s(X)} \lsquare{\lrparen{\EPt{\R(\Q,w,-Z)}{Z}{t} + G_t(v(\Q,w,-Z))} \cap M_t}
            \\&\quad\quad\quad\quad\quad\quad\quad\quad\quad\quad \rsquare{-^. \beta_t(\R(\Q,w,-Z),v(\Q,w,-Z)) }\\
            &\subseteq \bigcap_{(\Q,w) \in \W_t} \cl \bigcup_{Z \in R_s(X)} \lrsquare{\lrparen{\EQt{Z}{t} + G_t(w)} \cap M_t -^. \beta_{t,s}(\Q,w)}.
        \end{align*}
    \end{proof}

\subsection{Proof of Lemma~\ref{lemma2}}

\begin{proof}
    First, note that inclusion~\eqref{eq_supermtg-1} is equivalent to inclusion~\eqref{eq_supermtg-acceptance-1} by \cite[Lemma 3.6(iii)]{FR12}.
    Second, we will show that inclusion~\eqref{eq_supermtg-acceptance-2} is equivalent to 
        \begin{equation}\label{eq_l2}
        R_t(X) \subseteq \bigcap_{(\Q,w) \in \W_t} \bigcup_{Z \in R_s(X)} \lrsquare{\lrparen{\EQt{Z}{t} + G_t(w)} \cap M_t -^. \beta_{t,s}(\Q,w)}
        \end{equation}
        for all $X \in \LdpF{}$. Let $(\Q,w) \in \W_t$.
        To do this, we first show that \eqref{eq_l2} implies \eqref{eq_supermtg-acceptance-2}: Let $X \in A_t$.  By assumption,
            \begin{align*}
            0 &\in \bigcup_{Z \in R_s(X)} \lrsquare{\lrparen{\EQt{Z}{t} + G_t(w)} \cap M_t -^. \beta_{t,s}(\Q,w)}\\
            &= \lrcurly{m \in M_t \; | \; \exists Z \in R_s(X): \; \E{\trans{w}m} \geq \E{\trans{w}\EQt{Z}{t}} + \inf_{Y \in A_{t,s}} \E{\trans{w}\EQt{Y}{t}}}.
            \end{align*}
            This is true if and only if there exists a $Z \in R_s(X)$ such that \[\E{\trans{w}\EQt{-Z}{t}} \geq \inf_{Y \in A_{t,s}} \E{\trans{w}\EQt{Y}{t}},\]
            i.e., $-R_s(X) \cap \cl\lrparen{A_{t,s} + G_s^M(w_t^s(\Q,w))}\neq\emptyset$.  By \cite[Lemma 3.6(i)]{FR12} this is true if and only if $X \in A_s + \cl\lrparen{A_{t,s} + G_s^M(w_t^s(\Q,w))}$.  
        
        Now, we show that \eqref{eq_supermtg-acceptance-2} implies \eqref{eq_l2}: Let $m \in R_t(X)$. Then, $X + m \in A_s + \cl\lrparen{A_{t,s} + G_s^M(w_t^s(\Q,w))}$ by assumption.  By the equivalences shown above, and translativity, 
        \[m \in \bigcup_{Z \in R_s(X)} \lrsquare{\lrparen{\EQt{Z}{t} + G_t(w)} \cap M_t -^. \beta_{t,s}(\Q,w)}.\]
    The last step of the proof is to show that one can remove the closure in \eqref{eq_supermtg-2}, which reduces then to \eqref{eq_l2}.  This can be done using that $R_s$ is c.u.c.\ and convex.  For notation let $\hat{R}_{t,s}^{(\Q,w)}(Z) := \lrparen{\EQt{-Z}{t} + G_t(w)} \cap M_t -^. \beta_{t,s}(\Q,w)$, which is a convex risk measure for any $(\Q,w) \in \W_t$.  Then $\tilde R_t(X) := \bigcup_{Z \in R_s(X)} \hat{R}_{t,s}^{(\Q,w)}(-Z)$ is closed (for any choice $(\Q,w) \in \W_t$) if $\tilde A_t := \lrcurly{X \in \LdpF{} \;|\; 0 \in \tilde{R}_t(X)}$ is closed.  Using the proof of \cite[Lemma B.2]{FR12b}, $\tilde A_t = \tilde R_t^{-1}[M_{t,-}]$, which is closed if $\hat R_{t,s}^{(\Q,w),-1}[M_{t,-}]$ is a closed, convex, upper set (as $R_s$ is c.u.c.\ and convex). But this is true as 
        \begin{align*}
        \hat R_{t,s}^{(\Q,w),-1}[M_{t,-}] &= \lrcurly{Z \in M_s \; | \; 0 \geq \inf_{Y \in A_{t,s}} \E{\trans{w}\EQt{Y}{t}} + \E{\trans{w}\EQt{-Z}{t}}}\\
            &= \lrcurly{Z \in M_s \; | \; \E{\trans{w_t^s(\Q,w)}Z} \geq \inf_{Y \in A_{t,s}} \E{\trans{w_t^s(\Q,w)}Y}}\\
            &= \cl\lrparen{A_{t,s} + G_s^M(w_t^s(\Q,w))},
        \end{align*}
        which is a closed, convex, and upper set.
\end{proof}

\subsection{Proof of Corollary~\ref{cor_mtg}}
\label{sec_proof_cor_mtg}

\begin{proof}
Fix $X \in \LdpF{}$ and let $(\Q,w) \in \W_0$ such that $\beta_0(\Q,w) \neq \emptyset$ and
\[\cl\lrsquare{R_0(X) + G_0^M(w)} = \lrparen{\EQ{-X} + G_0(w)}\cap M_0 -^. \beta_0(\Q,w).\]
Define $m_t^X \in M_t$ such that
\[m_t^X + G_t^M(w_0^t(\Q,w)) = \lrparen{\EQt{X}{t} + G_t(w_0^t(\Q,w))} \cap M_t\]
and let $U_t := V_t^{(\Q,w_0^t(\Q,w))}(X) + m_t^X$ for any time $t$.  We claim that $U_t = G_t^M(w_0^t(\Q,w))$ for any time $t$, which we will then use to prove that $V_t^{(\Q,w_0^t(\Q,w))}(X)$ is a $\Q-$martingale.
Let us first show that $U_t \supseteq G_t^M(w_0^t(\Q,w))$:
    \begin{align}
    \nonumber &\inf_{u \in U_t} \E{\trans{w_0^t(\Q,w)} u} = \inf_{v \in V_t^{(\Q,w_0^t(\Q,w))}(X)} \E{\trans{w_0^t(\Q,w)} v} + \E{\trans{w_0^t(\Q,w)}m_t^X}\\
     \label{mtg-pf-supermtg}  &= \inf_{v \in V_t^{(\Q,w_0^t(\Q,w))}(X)} \trans{w}\EQ{v} + \trans{w}\EQ{X}
   \leq \inf_{v \in V_0^{(\Q,w)}} \trans{w}v + \trans{w}\EQ{X}\\
    \nonumber &= \inf_{Z \in R_0(X)} \trans{w}Z + \sup_{Y \in A_0} \trans{w}\EQ{-Y} + \trans{w}\EQ{X}\\
    \label{mtg-pf-wc} &= \trans{w}\EQ{-X} + \inf_{Y \in A_0} \trans{w}\EQ{Y} + \sup_{Y \in A_0} \trans{w}\EQ{-Y} + \trans{w}\EQ{X} = 0.
    \end{align}
    Inequality~\eqref{mtg-pf-supermtg} follows from the supermartingale property of Theorem~\ref{thm_supermtg} and \eqref{mtg-pf-wc} follows from the choice of $(\Q,w)$.
We now show that $U_t \subseteq G_t^M(w_0^t(\Q,w))$: 
    \begin{align}
    \nonumber &U_t = \cl\lrsquare{R_t(X) + \beta_t(\Q,w_0^t(\Q,w)) + m_t^X}\\
    \nonumber &= \cl\lrsquare{\lrparen{R_t(X) + G_t^M(w_0^t(\Q,w))} + \lrparen{m_t^X + \beta_t(\Q,w_0^t(\Q,w))}} \\
    \label{mtg-pf-dblsubtract} &= \cl\lrsquare{\lrparen{R_t(X) + G_t^M(w_0^t(\Q,w))} -^. \lrparen{G_t^M(w_0^t(\Q,w)) -^. \lrsquare{m_t^X + \beta_t(\Q,w_0^t(\Q,w))}}} \\
    \nonumber &= \cl\lrsquare{\lrparen{R_t(X) + G_t^M(w_0^t(\Q,w))} -^. \lrparen{\lrsquare{-m_t^X + G_t^M(w_0^t(\Q,w))} -^. \beta_t(\Q,w_0^t(\Q,w))}} \\
    \nonumber &= \cl\lsquare{\lrparen{R_t(X) + G_t^M(w_0^t(\Q,w))} -^. }
    \\
    \nonumber &\rsquare{\quad\quad\quad\quad\lrparen{\lrsquare{\EQt{-X}{t} + G_t(w_0^t(\Q,w))} \cap M_t -^. \beta_t(\Q,w_0^t(\Q,w))}} \\
    \label{mtg-pf-G} &\subseteq G_t^M(w_0^t(\Q,w)).
    \end{align}
    Equation~\eqref{mtg-pf-dblsubtract} follows from Proposition~\ref{prop_dblesubtract} and inclusion~\eqref{mtg-pf-G} follows from Proposition~\ref{prop_subtract-G}.
By definition of $U_t$, and given that $U_t = G_t^M(w_0^t(\Q,w))$, we obtain $V_t^{(\Q,w_0^t(\Q,w)}(X) = -m_t^X + G_t^M(w_0^t(\Q,w)) = \lrparen{\EQt{-X}{t} + G_t(w_0^t(\Q,w))} \cap M_t$ immediately implying by \cite[Corollary~A.4]{FR12b} that $\lrparen{V_t^{(\Q,w_0^t(\Q,w))}(X)}_{t = 0}^T$ is a $\Q-$martingale.  It remains to show that $(\Q,w_0^t(\Q,w))$ is a ``worst-case'' dual pair for any time $t$.  This follows from
\begin{align*}
\cl\lrsquare{R_t(X) + G_t^M(w_0^t(\Q,w))} &= V_t^{(\Q,w_0^t(\Q,w))}(X) -^. \beta_t(\Q,w_0^t(\Q,w))\\ 
&= \lrparen{\EQt{-X}{t} + G_t(w_0^t(\Q,w))} \cap M_t -^. \beta_t(\Q,w_0^t(\Q,w).
\end{align*}
The first equality follows from Proposition 2.4~(e1) and (e2) from~\cite{HS12} by noting $\beta_t(\Q,w_0^t(\Q,w)) \neq M_t$ and $\beta_t(\Q,w_0^t(\Q,w)) = \emptyset$ would imply $\beta_0(\Q,w) = \emptyset$ by multiportfolio time consistency (see Theorem~\ref{thm_penalty}), which would violate our assumption.

For the converse let $M = \R^d$, if $\lrparen{V_t^{(\Q,w_0^t(\Q,w)}(X)}_{t = 0}^T$ is a $\Q-$martingale for some $(\Q,w) \in \W_0$ with $\beta_0(\Q,w) \neq \emptyset$ then the process defined by $U_t := V_t^{(\Q,w_0^t(\Q,w)}(X) + \EQt{X}{t}$ is one as well.  In particular, at the terminal time $T$, $R_T(X) = R_T(0) - X$ and $U_T = \cl\lrsquare{R_T(0) + \beta_T(\Q,w_0^T(\Q,w))} = G_T(w_0^T(\Q,w))$ by $A_T = R_T(0)$ is a closed and conditionally convex cone (by closed, conditionally convex, and normalized).  Since $(U_t)_{t = 0}^T$ is a martingale this immediately implies $U_t = G_t(w_0^t(\Q,w))$ by \cite[Corollary~A.4]{FR12b}.  Therefore $\cl\lrsquare{R_t(X) + G_t(w_0^t(\Q,w))} = \lrparen{\EQt{-X}{t} + G_t(w_0^t(\Q,w))} -^. \beta_t(\Q,w_0^t(\Q,w))$ for any time $t$ (utilizing Proposition 2.4(e1) and (e2) from~\cite{HS12}).
\end{proof}

\begin{proposition}\label{prop_dblesubtract}
Let $A,B \in \mathcal{G}(M_t;M_{t,+})$ and $w \in \plus{M_{t,+}}\backslash\prp{M_t}$. Then, \[\cl\lrparen{A + B + G_t^M(w)} = \cl\lrparen{A + G_t^M(w)} -^. \lrparen{G_t^M(w) -^. B}.\]
\end{proposition}
\begin{proof}
\begin{align*}
&\cl\lrparen{A + G_t^M(w)} -^. \lrparen{G_t^M(w) -^. B} = \lrcurly{m \in M_t \; | \; m + \lrparen{G_t^M(w) -^. B} \subseteq \cl\lrparen{A + G_t^M(w)}}\\
&= \lrcurly{m \in M_t \; | \; m + \lrcurly{n \in M_t \; | \; n + B \subseteq G_t^M(w)} \subseteq \cl\lrparen{A + G_t^M(w)}}\\
&= \lcurly{m \in M_t \; | \; \E{\trans{w}m} +}
\\
&\quad\quad \rcurly{\inf\lrcurly{\E{\trans{w}n} \; | \; n \in M_t, \; \E{\trans{w}n} + \inf_{b \in B} \E{\trans{w}b} \geq 0} \geq \inf_{a \in A} \E{\trans{w}a}}\\
&= \lrcurly{m \in M_t \; | \; \E{\trans{w}m} - \inf_{b \in B} \E{\trans{w}{b}} \geq \inf_{a \in A} \E{\trans{w}a}}\\
&= \lrcurly{m \in M_t \; | \; \E{\trans{w}m} \geq \inf_{c \in A + B} \E{\trans{w}c}} = \cl\lrparen{A + B + G_t^M(w)}.
\end{align*}
\end{proof}

\begin{proposition}\label{prop_subtract-G}
Let $A,B \in \mathcal{G}(M_t;M_{t,+})$ and $w \in \plus{M_{t,+}}\backslash\prp{M_t}$.  If $\cl\lrparen{A + G_t^M(w)} \subseteq \cl\lrparen{B + G_t^M(w)}$ then $\cl\lrparen{A + G_t^M(w)} -^. \cl\lrparen{B + G_t^M(w)} \subseteq G_t^M(w)$.
\end{proposition}
\begin{proof}
\begin{align*}
\cl\lrparen{A + G_t^M(w)} &-^. \cl\lrparen{B + G_t^M(w)} = \lrcurly{m \in M_t \; | \; m + B \subseteq \cl\lrparen{A + G_t^M(w)}}\\
&= \lrcurly{m \in M_t \; | \; \E{\trans{w}m} + \inf_{b \in B} \E{\trans{w}b} \geq \inf_{a \in A} \E{\trans{w}a}}\\
&= \lrcurly{m \in M_t \; | \; \E{\trans{w}m} \geq \inf_{a \in A} \E{\trans{w}a} - \inf_{b \in B} \E{\trans{w}b}}\\
&\subseteq \lrcurly{m \in M_t \; | \; \E{\trans{w}m} \geq 0} = G_t^M(w).
\end{align*}
\end{proof}

\section{Proofs for Section~\ref{sec_supermtg-cc}}

\subsection{Proof of Corollary~\ref{cor_supermtg-cc}}

\begin{proof}
We will prove this by showing that the conditional supermartingale property is equivalent to the inclusions
    \begin{align}
        \label{eq_supermtg-1-cc} R_t(X) &\supseteq \bigcup_{Z \in R_s(X)} R_t(-Z)\\
        \label{eq_supermtg-2-cc} R_t(X) &\subseteq \bigcap_{(\Q,w) \in \W_t} \cl \bigcup_{Z \in R_s(X)} \lrsquare{\lrparen{\EQt{Z}{t} + \Gamma_t(w)} \cap M_t -^. \alpha_{t,s}(\Q,w)}.
    \end{align}
Then we will show that \eqref{eq_supermtg-1-cc} and \eqref{eq_supermtg-2-cc} are equivalent to multiportfolio time consistency.

The first part of this proof will be accomplished similarly to Lemma~\ref{lemma1}.  We will focus on certain points that are nontrivial to prove. We start with the proof that \eqref{eq_supermtg-1-cc} and \eqref{eq_supermtg-2-cc} imply the supermartingale property.
First, we can see that for $(\Q,w) \in \W_t^t$ (see Proposition~\ref{prop_cond_Wts}), Proposition~\ref{prop_essinf_rep} yields
\begin{align*}
\V_t^{(\Q,w)}(X) &= \lrcurly{u \in M_t\;|\; \trans{w}u \geq \essinf_{Z \in R_t(X)} \trans{w}Z + \essinf_{a_t \in \alpha_t(\Q,w)} -\trans{w}a_t}\\
&= \lrcurly{u \in M_t\;|\; \trans{w}u \geq \essinf_{Z \in R_t(X)} \trans{w}Z + \esssup_{Y \in A_t} \trans{w}\EQt{-Y}{t}}.
\end{align*}
Now we will show that
\begin{align*}
&\cl\EQt{\V_s^{(\Q,w_t^s(\Q,w))}(X)}{t} = 
\\
&\quad\quad\lrcurly{u \in M_t \; | \; \trans{w}u \geq \essinf_{Z \in R_s(X)} \trans{w}\EQt{Z}{t} + \esssup_{Y\in A_s} \trans{w}\EQt{-Y}{t}}.
\end{align*}

``$\subseteq$'' The right hand side is closed by the same logic as in Proposition~\ref{prop_essinf_rep}. Furthermore,
\begin{align*}
&\EQt{\V_s^{(\Q,w_t^s(\Q,w))}(X)}{t} = \lcurly{\EQt{u_s}{t} \; | \; u_s \in M_s, }\\
&\quad\quad \rcurly{ \trans{w_t^s(\Q,w)}u_s \geq \essinf_{Z \in R_s(X)} \trans{w_t^s(\Q,w)}Z + \esssup_{Y \in A_s} \trans{w_t^s(\Q,w)}\EQt{-Y}{s}}\\
&\subseteq \lcurly{\EQt{u_s}{t} \; | \; u_s \in M_s, \; \Et{\trans{w_t^s(\Q,w)}u_s}{t} \geq }\\
& \quad\quad \rcurly{
\Et{\essinf_{Z \in R_s(X)} \trans{w_t^s(\Q,w)}Z}{t} + \Et{\esssup_{Y \in A_s} \trans{w_t^s(\Q,w)}\EQt{-Y}{s}}{t}}\\
&= \lrcurly{u_t \in M_t \; | \; \trans{w}u_t \geq \essinf_{Z \in R_s(X)} \trans{w}\EQt{Z}{t} + \esssup_{Y \in A_s} \trans{w}\EQt{-Y}{t}}.
\end{align*}
Note that we are able to interchange the essential infimum/supremum and conditional expectation due to the decomposability property of $R_s(X)$ and $A_s$.

``$\supseteq$'' By way of contradiction, assume 
$m$ is an element of the right hand side
and $m \not\in \cl\EQt{\V_s^{(\Q,w_t^s(\Q,w))}(X)}{t}$.  Since $\cl\EQt{\V_s^{(\Q,w_t^s(\Q,w))}(X)}{t}$ is closed and convex, we can separate $\{m\}$ from it by some $v \in \LdqF{t}$.  That is
\begin{align*}
\E{\trans{v}m} &< \inf_{u_t \in \cl\EQt{\V_s^{(\Q,w_t^s(\Q,w))}(X)}{t}} \E{\trans{v}u} = \inf_{u_s \in \V_s^{(\Q,w_t^s(\Q,w))}(X)} \E{\trans{v}\EQt{u_s}{t}}\\
&= \E{\essinf_{u_s \in \V_s^{(\Q,w_t^s(\Q,w))}(X)} \trans{w_t^s(\Q,v)}u_s},
\end{align*}
where in the last equality above we can interchange the expectation and infimum since $\V_s^{(\Q,w_t^s(\Q,w))}(X)$ is decomposable.  By construction
\[
\essinf_{u_s \in \V_s^{(\Q,w_t^s(\Q,w))}(X)}\trans{w_t^s(\Q,v)}u_s 
= 
\begin{cases}
\begin{array}{l}
\essinf_{Z \in R_s(X)}\trans{w_t^s(\Q,v)}Z 
\\ 
\quad+ \esssup_{Y \in A_s}\trans{w_t^s(\Q,v)}\EQt{-Y}{s} \end{array}
& \text{on } D
\\ 
-\infty 
&\text{on } D^c,\end{cases}\]
where $D = \lrcurly{\omega \in \Omega \; | \; G_0^M(w_t^s(\Q,v)[\omega]) = G_0^M(w_t^s(\Q,w)[\omega])}$.  Since $\Q \sim \P$, we can conclude that $v(\omega) = \lambda(\omega)w(\omega)$ for some $\lambda \in \LpK{0}{t}{\R_{++}}$ (such that $\lambda w \in \LdqF{t}$).  Therefore $\E{\essinf_{u_s \in \V_s^{(\Q,w_t^s(\Q,w))}(X)} \trans{w_t^s(\Q,v)}u_s} > -\infty$ if and only if
\begin{align*}
&\E{\essinf_{u_s \in \V_s^{(\Q,w_t^s(\Q,w))}(X)} \trans{w_t^s(\Q,v)}u_s} = \E{\lambda \essinf_{u_s \in \V_s^{(\Q,w_t^s(\Q,w))}(X)} \trans{w_t^s(\Q,w)}u_s}\\
 =\;&\E{\lambda \lrparen{\essinf_{Z \in R_s(X)} \trans{w_t^s(\Q,w)}Z + \esssup_{Y \in A_s} \trans{w_t^s(\Q,w)}\EQt{-Y}{s}}}\\
=\;&\E{\lambda \lrparen{\essinf_{Z \in R_s(X)} \trans{w}\EQt{Z}{t} + \esssup_{Y\in A_s} \trans{w}\EQt{-Y}{t}}}.
\end{align*}
But this implies \[\E{\lambda \trans{w}m} < \E{\lambda \lrparen{\essinf_{Z \in R_s(X)} \trans{w}\EQt{Z}{t} + \esssup_{Y \in A_s} \trans{w}\EQt{-Y}{t}}},\] which is a contradiction to 
\[m \in \lrcurly{u \in M_t \; | \; \trans{w}u \geq \essinf_{Z \in R_s(X)} \trans{w}\EQt{Z}{t} + \esssup_{Y \in A_s} \trans{w}\EQt{-Y}{t}}.\]
The remaining part of the proof that inclusions~\eqref{eq_supermtg-1-cc} and~\eqref{eq_supermtg-2-cc} imply the conditional supermartingale property is similar to the proof of Lemma~\ref{lemma1}.

For the reverse implication, let us now assume the conditional supermartingale property.  We will prove inclusion \eqref{eq_supermtg-1-cc} by showing that the conditional supermartingale property implies time consistency, which then yields~\eqref{eq_supermtg-1-cc}.  This proof follows in total analogy to the corresponding proof in Lemma~\ref{lemma1} since
\[R_t(X) = \bigcap_{(\Q,w) \in \W_t^t} \lrsquare{\lrparen{\EQt{-X}{t} + \Gamma_t(w)}\cap M_t -^. \alpha_t(\Q,w)}.\]
Then, one shows that inclusion \eqref{eq_supermtg-2-cc} holds, which follows from the same logic as Lemma~\ref{lemma1} but using the scalarization results from Section~\ref{sec_condscal} instead. Let us give a short summary of the steps involved.

        \cite[Lemma 3.18]{FR13-survey} provides a representation for set-valued dynamic risk measure as an intersection of conditional scalarizations, where one can restrict to $w \in \plus{\recc{R_t(0)}}$ 
        \begin{align*}
        R_{t,s}(Z) &= \bigcap_{w \in \plus{\recc{R_t(0)}} \backslash \prp{M_t}} \lrcurly{m \in M_t \; | \; \trans{w}m \geq \hat\rho_{t,s}(Z)}\\
        &\subseteq 
        \bigcap_{w \in \plus{\recc{R_t(0)}} \backslash \prp{M_t}} \lrcurly{m \in M_t \; | \; \trans{w}m \geq \hat\rho_t(Z)}
        = R_t(Z) = R_{t,s}(Z).
        \end{align*}
        Now using the dual representations of the conditional scalarizations $\hat\rho_t(Z)$ and $\hat\rho_{t,s}(Z)$ from Proposition~\ref{prop_scalar-cc} and Corollary~\ref{cor_scalar_stepped-cc}, we obtain the following since the inclusion above is in fact an equality.
        For every $Z \in M_s$ and every $w \in \plus{\recc{R_t(0)}}\backslash \prp{M_t}$ it holds that for all $(\Q,m_{\perp}) \in \W_t(w)$ 
        \[\essinf_{Y_{t,s} \in A_{t,s}} \transp{w+m_{\perp}}\EQt{Y_{t,s} - Z}{t} \leq \esssup_{(\R,n_{\perp}) \in \W_t(w)} \essinf_{Y_t \in A_t} \transp{w+n_{\perp}}\ERt{Y_t-Z}{t}.\]
        This is because every such constraint is ``active'' in the above intersection, i.e., if any were made any stricter it would shrink the set $R_{t,s}(Z)$.
        Because of decomposability of 
        \[\lrcurly{\essinf_{Y \in A_t} \transp{w+m_{\perp}}\EQt{Y - X}{t} \; | \; (\Q,m_{\perp}) \in \W_t(w)},\]
         there exists a monotonically increasing sequence that converges to the essential supremum almost surely.  
        Now consider the above inequality for the case $m_{\perp} = 0 \in \prp{M_t}$ and also note that for $w \not\in \plus{\recc{R_t(0)}} \backslash \prp{M_t}$, $\essinf_{Y_{t,s} \in A_{t,s}} \transp{w+m_{\perp}}\EQt{Y_{t,s}-Z}{t} = -\infty$ 
        for any $(\Q,m_{\perp}) \in \W_t(w)$.
        Then, we can conclude for every $Z \in M_s$ it holds that for all $ (\Q,w) \in \W_t$ there exists a sequence $(\R_k,v_k)_{k \in \mathbb{N}} \subseteq \W_t$ such that $v_k \in w + \prp{M_t}$ for every $k$ and
        \[\essinf_{Y_{t,s} \in A_{t,s}} \trans{w}\EQt{Y_{t,s} - Z}{t} \leq \overline\lim_{k \to \infty} \essinf_{Y_t \in A_t} \trans{v_k}\EPt{\R_k}{Y_t - Z}{t},\]
        where $\overline\lim$ indicates the almost sure limit.
        In particular, for every $(\Q,w) \in \W_t$ there exists some sequence $\R_k(\Q,w,Z) \in \mathcal{M}$ and $v_k(\Q,w,Z) \in w + \prp{M_t}$ such that $(\R_k(\Q,w,Z),v_k(\Q,w,Z))_{k \in \mathbb{N}} \subseteq \W_t$ and
         \begin{align*}
          &\lrparen{\EQt{-Z}{t} + \Gamma_t(w)}\cap M_t -^. \alpha_{t,s}(\Q,w) \supseteq\\
          & \overline\cl\bigcup_{k \in \mathbb{N}} \lrsquare{\lrparen{\EPt{\R_k(\Q,w,Z)}{-Z}{t} + \Gamma_t(v_k(\Q,w,Z))} \cap M_t -^. \alpha_t(\R_k(\Q,w,Z),v_k(\Q,w,Z))},
          \end{align*}
        where $\overline\cl$ indicates the almost sure closure.
        This implies the desired inclusion \eqref{eq_supermtg-2-cc}:
        \begin{align*}
        R_t(X) &\subseteq \bigcap_{(\Q,w) \in \W_t} \cl \bigcup_{Z \in R_s(X)} \lrsquare{\lrparen{\EQt{Z}{t} + \Gamma_t(w)} \cap M_t -^. \alpha_t(\Q,w)}\\
        &\subseteq \bigcap_{(\Q,w) \in \W_t} \cl \bigcup_{Z \in R_s(X)} \overline\cl\bigcup_{k \in \mathbb{N}} \lsquare{\lrparen{\EPt{\R_k(\Q,w,-Z)}{Z}{t} + \Gamma_t(v_k(\Q,w,-Z))} \cap M_t}
            \\&\quad\quad\quad\quad\quad\quad\quad\quad\quad\quad \rsquare{-^. \alpha_t(\R_k(\Q,w,-Z),v_k(\Q,w,-Z)) }\\
            &\subseteq \bigcap_{(\Q,w) \in \W_t} \cl \bigcup_{Z \in R_s(X)} \lrsquare{\lrparen{\EQt{Z}{t} + \Gamma_t(w)} \cap M_t -^. \alpha_{t,s}(\Q,w)}.
        \end{align*}

The last part of the proof is to show that \eqref{eq_supermtg-1-cc} and \eqref{eq_supermtg-2-cc} are equivalent to multiportfolio time consistency.
It is trivially true that multiportfolio time consistency implies the inclusions \eqref{eq_supermtg-1-cc} and \eqref{eq_supermtg-2-cc} by the recursive formulation in Theorem~\ref{thm_equiv_tc}, using for \eqref{eq_supermtg-2-cc} that the union of intersections is contained in intersection of unions. 
For the converse implication we can use the same logic as in the proof of Lemma~\ref{lemma2} to show that \eqref{eq_supermtg-1-cc} and \eqref{eq_supermtg-2-cc} are equivalent to
\begin{align}
    \label{a}A_t &\supseteq A_s + A_{t,s}\\
    \label{b}A_t &\subseteq \bigcap_{(\Q,w) \in \W_t} \lrsquare{A_s + \cl\lrparen{A_{t,s} + \Gamma_s^M(w_t^s(\Q,w))}}
\end{align}
when $\seq{R}$ is conditionally c.u.c.
By $\Gamma_s^M(w_t^s(\Q,w)) \subseteq G_s^M(w_t^s(\Q,w))$,  \eqref{a} and \eqref{b} imply \eqref{eq_supermtg-acceptance-1} and \eqref{eq_supermtg-acceptance-2}, which imply multiportfolio time consistency by Theorem~\ref{thm_supermtg} and Lemmas~\ref{lemma1} and \ref{lemma2}.
\end{proof}

\subsection{Proof of Corollary~\ref{cor_mtg-cc}}

\begin{proof}
Most aspects of the proof are similar to the proof of Corollary~\ref{cor_mtg}, using a straight forward extension of Proposition 2.4~(e1) and (e2) from~\cite{HS12} for conditionally convex and decomposable sets and the trivial conditional versions of Proposition~\ref{prop_dblesubtract} and~\ref{prop_subtract-G} (given in Proposition~\ref{prop_gamma_subtract}). 
We will here only show $\U_t =\Gamma_t^M(w_0^t(\Q,w))$ as the proof differs slightly from the proof of Corollary~\ref{cor_mtg}. Define $m_t^X \in M_t$ such that
\[m_t^X + \Gamma_t^M(w_0^t(\Q,w)) = \lrparen{\EQt{X}{t} + \Gamma_t(w_0^t(\Q,w))} \cap M_t\]
and define $\U_t := \V_t^{(\Q,w_0^t(\Q,w))}(X) + m_t^X$ for any time $t$. Since $\Gamma_0^M = G_0^M$ by definition, one obtains $\U_0 = \Gamma_0^M(w)$ from the proof of Corollary~\ref{cor_mtg}.  Similar to the proof of Corollary~\ref{cor_mtg}, we can show that $\U_t \subseteq \Gamma_t^M(w_0^t(\Q,w))$ for any time $t$.  For the reverse, we use the fact that $\U_t$ defines a supermartingale, which follows from the properties of $\V_t^{(\Q,w_0^t(\Q,w))}(X)$ and $m_t^X$. Thus, $\U_t \subseteq \cl\Et{\U_s}{t}$. Therefore, for any time $t$ we obtain
\[\Gamma_0^M(w) = \U_0 \subseteq \cl\EQ{\U_t} \subseteq \cl\EQ{\Gamma_t^M(w_0^t(\Q,w))} = \Gamma_0^M(w)\]
by \cite[Corollary~A.6]{FR12b}.  Let us assume $\U_t \subsetneq \Gamma_t^M(w_0^t(\Q,w))$, i.e.\ there exists some $\delta > 0$ such that $\P(\essinf_{u \in \U_t} \trans{w_0^t(\Q,w)}u \geq \delta) > 0$.  However, this contradicts $\cl\EQ{\U_t} = \Gamma_0^M(w)$ since
\[0 = \inf_{u_0 \in \cl\EQ{\U_t}} \trans{w}u = \inf_{u_t \in \U_t} \E{\trans{w_0^t(\Q,w)}u_t} \geq \delta \P(\essinf_{u_t \in \U_t} \trans{w_0^t(\Q,w)}u_t) > 0.\]
\end{proof}

\begin{proposition}\label{prop_gamma_subtract}
Let $A,B \in \mathcal{G}(M_t;M_{t,+})$ conditionally convex and $w \in \plus{M_{t,+}}\backslash\prp{M_t}$.  Then
\[\cl\lrparen{A + B + \Gamma_t^M(w)} = \cl\lrparen{A + \Gamma_t^M(w)} -^. \lrparen{\Gamma_t^M(w) -^. B}.\]
Additionally, if $\cl\lrparen{A + \Gamma_t^M(w)} \subseteq \cl\lrparen{B + \Gamma_t^M(w)}$ then \[\cl\lrparen{A + \Gamma_t^M(w)} -^. \cl\lrparen{B + \Gamma_t^M(w)} \subseteq \Gamma_t^M(w).\]
\end{proposition}

\bibliographystyle{plain}
\bibliography{biblio}
\end{document}